\documentclass[journal,letterpaper,twocolumn,twoside,nofonttune]{IEEEtran}
\usepackage{mathpazo}
\usepackage{times}

\usepackage[T1]{fontenc}
\usepackage{amsmath,amssymb,amsfonts}
\usepackage{mathrsfs}
\usepackage{mathabx}
\usepackage{amsbsy}
\usepackage{graphicx}
\usepackage{graphics}
\usepackage{epstopdf}
\usepackage{algorithm}
\usepackage{algorithmic}
\usepackage{subfigure}
\usepackage{cite}
\usepackage{multirow}

\usepackage{caption}
\usepackage{tikz,pgfplots}
\usepackage{ctable}
\usepackage{setspace}
\usepackage{mathtools}

\title{\Huge$\,$\\[-2.75ex]
{Analog Subspace Coding: A New Approach to Coding for Non-Coherent Wireless Networks}\\[0.50ex]}

\author{\large%
Mahdi Soleymani 
and 
Hessam Mahdavifar,\,\,\IEEEmembership{Member,~IEEE}\\
\vspace{-.25in}
\thanks{%
 The material in this paper was presented in part at the IEEE International Symposium on Information Theory in June 2020.
}
\thanks{This work was supported by the National Science Foundation
under grants CCF--1763348, CCF--1909771, and CCF--1941633.}
\thanks{M.\ Soleymani and H.\ Mahdavifar are with the Department of Electrical Engineering and Computer Science, University of Michigan, Ann Arbor, MI 48104 (email: mahdy@umich.edu and hessam@umich.edu).}
}

\newtheorem{theorem}{{Theorem}}
\newtheorem{lemma}[theorem]{{Lemma}}

\newtheorem{corollary}[theorem]{{Corollary}}

\newtheorem{definition}{{Definition}}

\newcommand{\cC}{{\cal C}}

\newcommand{\cF}{{\cal F}}
 
\newcommand{\cH}{{\cal H}}

\newcommand{\cN}{{\cal N}}
 
\newcommand{\cP}{{\cal P}}
 
\newcommand{\cR}{{\cal R}}

\newcommand{\cW}{{\cal W}}

\DeclareMathAlphabet{\mathbfsl}{OT1}{ppl}{b}{it} 

\newcommand{\bA}{\mathbfsl{A}} 
\newcommand{\bB}{\mathbfsl{B}}
\newcommand{\bC}{\mathbfsl{C}} 

\newcommand{\bE}{\mathbfsl{E}} 

\newcommand{\bG}{\mathbfsl{G}} 
\newcommand{\bH}{\mathbfsl{H}}
\newcommand{\bI}{\mathbfsl{I}}

\newcommand{\bN}{\mathbfsl{N}}
 
\newcommand{\bP}{\mathbfsl{P}}
\newcommand{\bQ}{\mathbfsl{Q}} 
\newcommand{\bR}{\mathbfsl{R}}
 
\newcommand{\bT}{\mathbfsl{T}}

\newcommand{\bX}{\mathbfsl{X}}
\newcommand{\bY}{\mathbfsl{Y}} 
\newcommand{\bZ}{\mathbfsl{Z}}

\newcommand{\bc}{\mathbfsl{c}} 
\newcommand{\bu}{\mathbfsl{u}} 
\newcommand{\bv}{\mathbfsl{v}}
 
\newcommand{\bx}{\mathbfsl{x}}
\newcommand{\by}{\mathbfsl{y}}

\newcommand{\floor}[1]{\left\lfloor #1 \right\rfloor}

\newcommand*{\rom}[1]{\expandafter\romannumeral #1}

\makeatletter
\newcommand{\AlignFootnote}[1]{%
  \ifmeasuring@
  \else
    \iffirstchoice@
      \footnote{#1}%
    \fi
  \fi}
\makeatother

\newcommand{\Fp}{{{\Bbb F}}_{\!p}}
\newcommand{\be}{\begin{equation}}
\newcommand{\ee}{\end{equation}} 
\newcommand{\eq}[1]{(\ref{#1})}

\renewcommand{\leq}{\leqslant}
 
\renewcommand{\geq}{\geqslant}

\renewcommand{\Bbb}{\mathbb}
\newcommand{\C}{{\Bbb C}} 
\newcommand{\N}{{\Bbb N}}
\newcommand{\R}{{\Bbb R}}

\newcommand{\Lf}{{\Bbb L}}

\newcommand{\Tref}[1]{Theo\-rem\,\ref{#1}}

\newcommand{\Lref}[1]{Lem\-ma\,\ref{#1}}
\newcommand{\Cref}[1]{Co\-ro\-lla\-ry\,\ref{#1}}
\newcommand{\Dref}[1]{Definition\,\ref{#1}}

\newcommand{\Fq}{{{\Bbb F}}_{\!q}}

\newcommand{\deff}{\mbox{$\stackrel{\rm def}{=}$}}

\DeclareMathOperator{\rank}{rank}

\newcommand{\Span}[1]{{\left\langle {#1} \right\rangle}}

\newcommand{\shalf}{\mbox{\raisebox{.8mm}{\footnotesize $\scriptstyle 1$}
\footnotesize$\!\!\! / \!\!\!$ \raisebox{-.8mm}{\footnotesize
$\scriptstyle 2$}}}

\newcommand{\dmin}{d_{\text{min}}}

\newcommand{\norm}[1]{\left\lVert#1\right\rVert}
\begin{document}

\maketitle

\begin{abstract}

We provide a novel framework to study subspace codes for non-coherent communications in wireless networks. To this end, an \textit{analog operator channel} is defined with inputs and outputs being subspaces of $\C^n$. Then a certain distance is defined to capture the performance of subspace codes in terms of their capability to recover from interference and rank-deficiency of the network. We also study the robustness of the proposed model with respect to an additive noise. Furthermore, we propose a new approach to construct subspace codes in the analog domain, also regarded as Grassmann codes, by leveraging polynomial evaluations over finite fields together with characters associated to finite fields that map their elements to the unit circle in the complex plane. The constructed codes, referred to as character-polynomial (CP) codes, are shown to perform better comparing to other existing constructions of Grassmann codes in terms of the trade-off between the rate and the normalized minimum distance, for a wide range of values for $n$.

\end{abstract}

\section{Introduction} 
\label{sec:Introduction}

 Wireless networks are rapidly growing in size, are becoming more hierarchical, and are becoming increasingly distributed. In the next generation of wireless cellular networks, namely 5G, tens of small cells, hundreds of mobile users demanding ultra-high data rates, and thousands of Internet-of-Things (IoT) devices will be all operating within the coverage of one single cell \cite{3gpp-5g-2}. While the efforts for 5G standardization are still ongoing, several new features have been introduced in the recent releases of the Long-Term Evolution (LTE) standard to start supporting the diverse requirements of the wide range of use cases in 5G. Started with Release 10 the deployment of small cells in LTE is becoming increasingly popular to deliver enhanced spectral capacity and extended network coverage \cite{nakamura2013trends}, which is also fundamental to enhanced mobile broadband (eMBB) and massive machine type communications (mMTC) scenarios in 5G. Moreover, features such as coordinated multipoint (CoMP) transmission and reception \cite{sun2013interference} together with enhanced intercell interference coordination (eICIC) \cite{deb2014algorithms} have been introduced and used since Release 10 and evolved since then. 
 
The aforementioned techniques are, however, difficult to scale as the number of small cells, that can be also regarded as relays, keeps increasing and as more layers are added in the hierarchical network. More specifically, conventional methods including channel estimation of point-to-point wireless links, link-level block coding, and successive interference cancellation do not properly scale with the size of such massive networks. Motivated by the emergence of such massive networks we study coding for wireless networks consisting of many relays operating in a non-coherent fashion, where the network nodes are oblivious to the channel gains of the point-to-point wireless links as well as the structure of the network. In a sense, this resembles a random linear network coding scenario, though completely in the physical layer, where physical-layer transport blocks are linearly combined in the relay nodes as they receive the spatial sum of blocks sent by the neighboring nodes. This holds assuming omni-directional radio frequency (RF) transmitter and receiver antennas are deployed at the network nodes. Also, in the considered setup, the relay nodes, such as small cells, do not attempt to decode messages and only amplify and forward the received physical-layer blocks. 

In this paper, we define a new framework for reliable communications over wireless networks in a non-coherent fashion, as discussed above, using \textit{analog subspace codes}. Let $\cW$ denote an ambient vector space of dimension $n$ over a field $\Lf$, i.e., $\cW = \Lf^n$. A subspace code in $\cW$ is a non-empty subset of the set of all the subspaces of $\cW$. We observe that subspace codes in the analog domain, where the underlying field $\Lf$ is $\R$ or $\C$, become relevant for conveying information across networks in such a scenario.

This work is mainly inspired by the seminal work by Koetter and Kschischang \cite{KK}, who defined a new framework for correcting errors and erasures in a randomized network coding scenario \cite{ho2003benefits}. They defined an \textit{operator channel} to capture the effect of errors and erasures in such a scenario and showed that subspace codes over finite fields are instrumental to provide reliability for communications over operator channels. In a sense, we develop a counterpart for Koetter-Kschischang's operator channel in the analog domain, referred to as \textit{analog operator channel}. More specifically, the analog operator channel models the \textit{rank-deficiency} of the network, caused by relay failures or lacking a sufficient number of active relays, as subspace erasures. Also, it models the interference from neighboring cells/small cells as subspace errors. We further discuss various methods for constructing subspace codes for the analog operator channel. In particular, we propose a novel construction method by leveraging characters associated to Abelian groups and finite fields, and mapping them to the unit circle in the complex plane. 

It is worth noting that the setup considered in this paper fundamentally differs from Koetter-Kschischang's setup in two main aspects. First, due to the fundamental differences between the structure of finite fields and the analog fields of $\R$ or $\C$ constructing codes for analog operator channels requires entirely different approaches comparing to subspace codes constructed over finite fields in \cite{KK}. Second, the effect of physical layer is abstracted out in the setup considered in \cite{KK} as it is often the case in the network coding literature. However, in this work, we arrive at the notion of analog operator channels of subspace codes with an innovative perspective, namely, physical layer communications over wireless networks. Hence, we consider the additive noise that is always present in the physical layer, in addition to subspace errors and erasures discussed above, and characterize the robustness of the analog operator channel model with respect to the additive noise.  

Analog subspace codes can be also viewed as codes in Grassmann space, also referred to as Grassmann codes, provided that the dimensions of all the subspace codewords are equal. There is a long history on studying bounds \cite{shannon1959probability,barg2002bounds,bachoc2006linear,bachoc2006bounds,barg2006bound}, using packing and covering arguments, and capacity analysis in Grassmann space, mostly motivated by space-time coding for multiple-input multiple-output (MIMO) wireless systems \cite{marzetta1999capacity, hochwald2000unitary, zheng2002communication}. In such systems, a separate block code is needed to guarantee the reliability regardless, and the space-time code can be interpreted as the means of improving the reliability by exploiting the diversity the MIMO channel offers. However, we arrive at the problem of constructing subspace codes from the analog operator channel. In other words, subspace codes are used for reliable communications over analog operator channels the same way block codes are conventionally used for reliable communications over point-to-point links. A more detailed overview of prior works on Grassmann codes and their relations to our approach is provided later in Section\,\ref{sec:sec2c}. 

The rest of this paper is organized as follows. In Section\,\ref{sec:two} the analog operator channel is defined and an overview of the related prior work on Grassmann codes is provided. In Section\,\ref{sec:three} a new notion of subspace distance is defined and its relation with correcting subspace errors and erasures is discussed. The robustness of the analog operator channel with respect to the additive noise is analyzed in Section\,\ref{sec:four}. Also, new constructions of analog subspace codes are discussed in Section\,\ref{sec:five}. Finally, the paper is concluded in Section \,\ref{sec:six}.

\section{Preliminaries}
\label{sec:two}

\subsection{Notation Convention}
\label{notation}

Let $[n]$ denote the set of positive integers less than or equal to $n$, i.e., $[n]=\{1,2,\hdots,n\}$ for $n \in \N$. Also, for $x \in \R$, $(x)_+\,\deff\, \max(0,x)$. 

In this paper, matrices are represented by bold capital letters. The row space of a matrix $\bX$ is denoted by $\Span{\bX}$. Also, for a square matrix $\bX$, the trace of $\bX$, denoted by $\text{tr}(\bX)$, is defined to be the sum of elements of $\bX$ on the main diagonal. 

The ambient vector space is denoted by $W$. The parameter $n$ is reserved for the dimension of $W$ throughout the paper. Also, we have $W = \Lf^n$, where $\Lf$ can be either $\R$ or $\C$. In order to state results for $\Lf$, which could be either $\R$ or $\C$, a parameter $\beta$ is defined, where $\beta=1$ for $\Lf=\R$ and $\beta = 2$ for $\Lf = \C$. Let $\cP(W)$ denote the set of all subspaces of $W$. For a subspace $V \in \cP(W)$, the dimension of $V$ is denoted by $\dim(V)$. The sum of  two subspaces $U, V \in \cP(W)$ is defined as
\be
\label{plus-def}
U+V\,\deff\,\{u+v  :  u\in U  ,  v\in V\}.
\ee
Note that if $U$ and $V$ intersect trivially, i.e., $U\cap V=\{ \boldsymbol{0}\}$, where $\boldsymbol{0}$ is the all-zero vector, then $U+V$ is a direct sum and is denoted by $U \oplus V$. 

The set of all $m$-dimensional subspaces of $\Lf^n$ is denoted by $G_{m,n}(\Lf)$, which is referred to as Grassmann space or Grassmannian in the literature. Given $\Lf = \C$, $G_{m,n}(\C)$ can be also described as follows:
\be\label{grassmanndef}
G_{m,n}(\C)\,\deff\,\{\Span{\bZ} : \bZ \in \C^{m\times n}, \bZ\bZ^{\text{H}}=\bI_m\},
\ee
where $\bI_m$ is the $m\times m$ identity matrix. The elements of $G_{m,n}(\Lf)$ are also referred to as $m$-planes. 

The Frobenius norm of a matrix $\bA$ is defined as 
\be \label{Frobnius norm deff}
\norm{\bA}\,\deff\,\sqrt{\text{tr}(\bA^{\text{H}}\bA)}=\sqrt{\text{tr}(\bA \bA^{\text{H}})}.
\ee

By fixing a basis for $W$, any vector in $W$ is represented by $n$-tuples of coordinates with respect to the chosen basis. The  inner product between $\bu,\bv\in W$ is then defined as: $\bu.\bv\,\deff\,\sum_{i=1}^n u_iv_i$. Then the orthogonal subspace of  $U\in \cP(W)$ is defined as
\be
\label{orthogonal subspace}
U^\perp\,\deff\,\{\bv \in W : \bu.\bv=0, \forall \bu \in U \}.
\ee

For a set $M$, a $\sigma$-quasimetric on $M$ is a  function $d:M \times M \xrightarrow{} \R$ that satisfies all the conditions of a metric except the triangle inequality being relaxed to 
\be\label{quasi-triangle-ineq}
\forall x,y,z \in M, \quad  d(x,z)<\sigma \bigl(d(x,y)+d(y,z)\bigr),
\ee
for a constant $\sigma >1$. This inequality is referred to as \emph{$\sigma$-relaxed triangle inequality}.

\subsection{Analog operator channel}
\label{sec:analog_operator}

This model is motivated by non-coherent communications over wireless networks, as discussed in Section\,\ref{sec:Introduction}. Hence, each piece of the model is followed by a brief explanation from this perspective. Let $\bx_i  \in \C^n$, for $i\in[m]$, denote the input vectors. The input vectors, as physical layer transport blocks, can be sent by several antennas of a transmitter, e.g., a cellular base station, at different time frames. By discarding the interference and the additive noise, the output of the channel is a set of vectors $\by_j=\sum_{i=1}^m h_{j,i} \bx_i$, where $j \in [l]$. Each vector $\by_j$ is the received transport block by an antenna of the receiver at a certain time frame. Note that a time-frame-level synchronization is assumed across the wireless links, e.g., by employing specific patterns in a designated subset of orthogonal frequency-division multiplexing (OFDM) symbols in each time frame as in LTE networks \cite{cox2012introduction}. Also, the relays in the network, e.g., small cells, are assumed to be amplify-and-forward relays. They can forward a transport block, received during a certain time frame, in a subsequent time frame. This is because the communication is assumed to be done in the unit of time frame, i.e., the relay has to wait for the current time frame to end before it can begin forwarding what it received. Then, due to the different delays, in the unit of time frames, that the transport blocks may encounter as they are propagated through the network, the received $\by_j$'s can be the combination of transmitted $\bx_i$'s across different antennas and time frames. Under a non-coherent scenario, both the transmitter and the receiver are oblivious to $h_{j,i}$'s, the topology of the network, and the link-level channel gains. It is possible that several interference blocks, e.g., up to $t$ of them, from neighboring cells/small cells are also received by the receiver. Hence, we have 
\be
\label{cheqs}
\bY_{l \times n}=\bH_{l\times m}\bX_{m \times n}+\bG_{l \times t}\bE_{t \times n},
\ee
where $\bX$'s rows are the transmitted blocks $\bx_1,\bx_2,\dots,\bx_m$,  $\bE$'s rows are the interference blocks $\boldsymbol{e}_1,\boldsymbol{e}_2,\dots,\boldsymbol{e}_t$, $\bY$'s rows are the received blocks $\by_1,\by_2,\dots,\by_l$, and $\bH = [h_{j,i}]_{l\times m}$ and $\bG = [g_{j,i}]_{l \times t}$ are assumed to be unknown to the transmitter and the receiver. Note that both $\bH$ and $\bG$ depend on the network topology as well as the link-level channel gains, however, $\bG$ also depends on the specific nodes where the interference blocks have entered the network. 

An example of the communication scenario, described by \eq{cheqs}, is illustrated in Figure\,1. Here, all the considered nodes have one transmit and one receive antennas and the communication is done in two time frames. We have $\bX=[\bx]_{1 \times n}$, 
$\bY=\left[
\begin{array}{c}
\by_1 \\
\by_2
\end{array}
\right]_{2 \times n}$, $\bE = [\bold{e}]_{1 \times n}$,
$\bG=\left[
\begin{array}{c}
g \\
0
\end{array}
\right]_{2 \times 1} $, and $\bH=\left[
\begin{array}{c}
h_1 \\
h_2h_3
\end{array}
\right]_{2 \times 1} $. In other words, the receiver receives $h_1 \bx + g \bold{e}$ in the first time frame and receives $h_2 h_3 \bx$ in the second time frame.

\begin{figure}[t]
\begin{center}
\includegraphics[width=.9\linewidth]{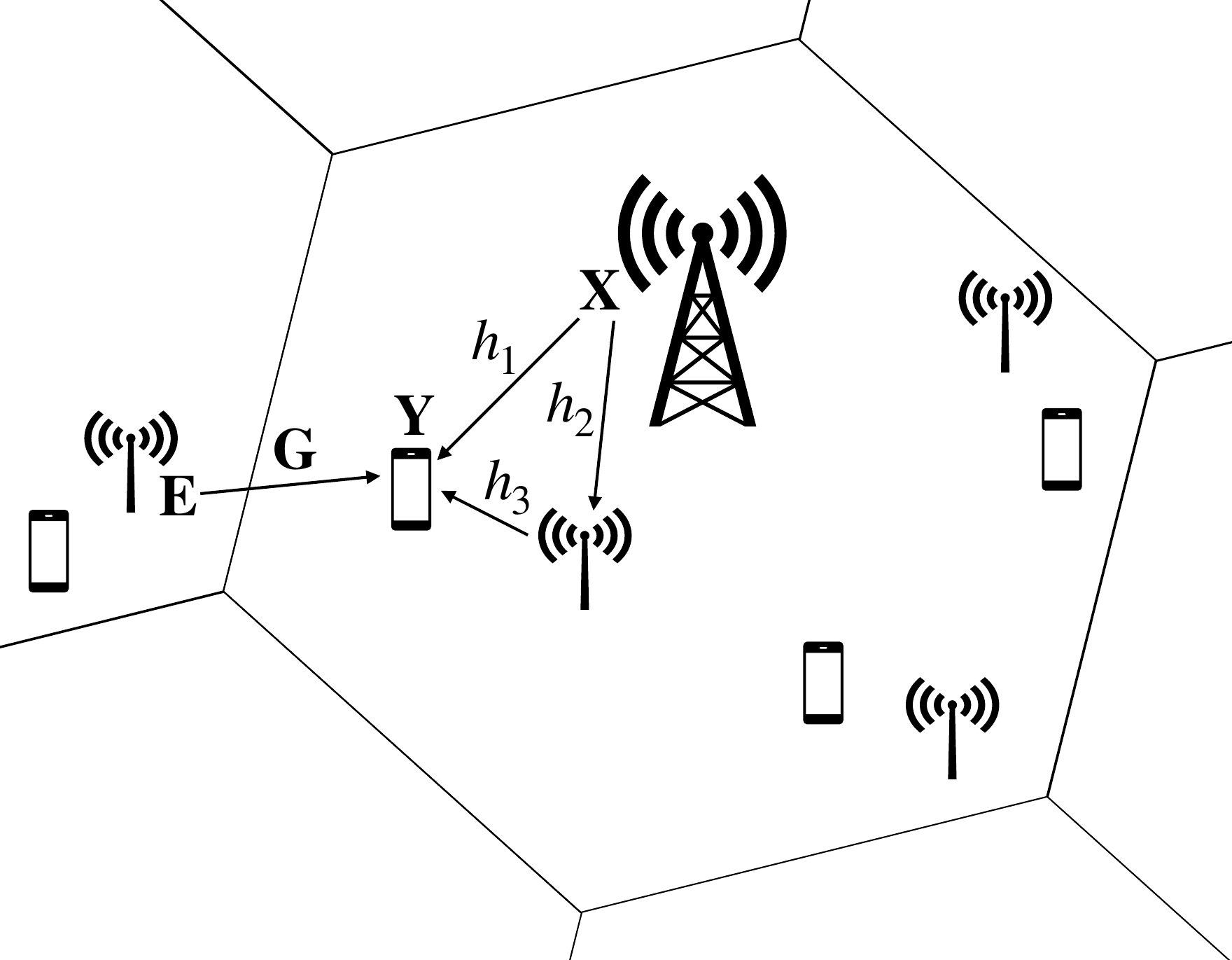}
\caption{\small An example of a non-coherent wireless network with the input-output relation specified in \eqref{cheqs}.}
\end{center}
\end{figure}

In the scenario described by \eq{cheqs}, even in the absence of the interference blocks $\bE$, the only way to convey information to the receiver is through the subspace spanned by the rows of $\bX$. This is mainly due to the underlying assumption on non-coherent communications, where $\bH$ is assumed to be completely unknown to both the transmitter and the receiver. Furthermore, $\bH$ may not be full column rank, e.g., when $l < n$, which implies that $\Span{\bX}$ can not be fully recovered. In order to capture the rank-deficiency of $\bH$, a stochastic erasure operator is defined as follows. For some $k\geq 0$, $\cH_k(U)$ returns a random $k$-dimensional subspace of $U$, if $\dim(U)> k$, and returns $U$ otherwise. Then the analog operator channel is defined as follows:

\begin{definition}\label{operatorchannel}
An analog operator channel associated with $W$ is a channel with input $U \in \cP(W)$ and output $V \in \cP(W)$ together with the following input-output relation:
\be
\label{channelIOrelation}
V=\cH_k(U)\oplus E,
\ee
where $E$ is the interference subspace, also referred to as the error subspace, with $E \cap U = \{\boldsymbol{0}\}$. Then $\rho = \text{dim}(U) -k$ is referred to as the dimension of erasures and $t = \text{dim}(E)$ is referred to as the dimension of errors. 
\end{definition}
\noindent
{\bf Remark\,1.} In the communication scenario described by \eq{cheqs}, the additive noise of the physical layer, often modeled as additive white Gaussian noise (AWGN), is discarded. Note that the intermediate relay nodes in the wireless network, such as small cells, are not often limited by power constraints as the end mobile users are. Hence, it is natural to assume that the relay nodes operate at high signal-to-noise ratio (SNR). Nevertheless, it is essential to investigate the effect of additive noise as a perturbation of the transformation described by \eq{cheqs}. In other words, instead of $\bY$, $\bY + \bN$ is received by the receiver, where $\bN$ is a matrix of i.i.d. Gaussian random variables. This, in turn, results in a perturbation in the analog operator channel, defined in Definition\,\ref{operatorchannel}. We will discuss the robustness of the considered channel model with respect to the additive noise in Section\,\ref{sec:four}.

\subsection{Related prior work}
\label{sec:sec2c}

The Grassmann space can be turned into a metric space using \textit{chordal} distance. Roughly speaking, \textit{chordal} distance generalizes the notion of angle between two lines to subspaces of equal dimension, which will be defined more precisely in the next section. Let $\delta_c$ denote the \textit{chordal} distance normalized by $n$. Also, let the rate $R$ of a code $\cC$ be defined as $\ln{|\cC|}/n$.

 The problem of deriving bounds on the minimum distance of subspace codes of fixed dimension, i.e., packing subspaces in $G_{m,n}(\R)$, was first studied by Shannon for the special case of $m=1$ \cite{shannon1959probability}. The packing problem in $G_{1,n}(\R)$ is also related to designing spherical codes, i.e., packing points on a hyper sphere in the Euclidean space with a given angular separation. A lower bound on the best rate $R$ of $\cC$ with the minimum angle $\theta$, i.e., $\delta_c = \sin(\theta)$, is derived by Shannon \cite{shannon1959probability} as $R>-\ln(\sin(\theta))$, assuming $n \rightarrow \infty$. Lower and upper bounds on the largest achievable rate $R$, given a fixed $\delta_c$ and $m$ while $n \rightarrow \infty$, were derived in \cite{barg2002bounds}. The upper bound was later improved in \cite{bachoc2006linear,bachoc2006bounds,barg2006bound}. 

In terms of lower bounds for the packing problem in the Grassmann space, an achievability bound for the minimum distance of the codes for finite values of $n$ was derived in \cite{henkel2005sphere}. The construction of subspace codes in $G_{m,n}(\R)$ and analysis of their minimum \textit{chordal} distance was also studied in \cite{conway1996packing}, and some of these constructions were observed to be optimal. However, the suggested construction methods in \cite{conway1996packing} are numerical, making them computationally infeasible for general parameters. Another numerical method for constructing codes based on alternative projection is proposed in \cite{dhillon2008constructing}. In another line of work, motivated by quantum error-correcting codes, constructions based on group structures are suggested \cite{nebe2001invariants,shor1998family,calderbank1999group,calderbank1997quantum,calderbank1998quantum} (see, e.g., \cite{sloane2002packing} for a brief survey). A connection between these codes and massive multiple access scenarios is observed in \cite{itw2019}, where low-complexity decoders are also proposed. Another related line of work is the \emph{frame} design problem in a Hilbert space, where a frame is a set of overcomplete unit norm vectors having small mutual inner product. Such a design has applications to a broad range of problems, including problems in signal processing, distributed sensing, parallel processing, etc \cite{strohmer2003grassmannian,xia2005achieving,tropp2005designing,kutyniok2009robust,ding2007generic}.

The most notable line of work on Grassmann codes is motivated by space-time coding for MIMO channels. In other words, the problem of constellation design for communications over a non-coherent MIMO is observed to be closely related to the packing problem in the complex Grassmann space. In this line of work, initiated by Marzetta and Hochwald \cite{marzetta1999capacity}, it is assumed that neither the transmitter nor the receiver knows the fading coefficients of the channel which are assumed to be fixed in the coherence time of the channel. They also proposed a constellation scheme in the Grassmann space called unitary space-time modulation \cite{hochwald2000unitary}. Zheng and Tse \cite{zheng2002communication} derived the capacity of non-coherent MIMO channel in high SNR in terms of the channel coherence time and the number of receive and transmit antennas. They further give a geometric interpretation of the capacity expression as sphere packing in the Grassmann manifold \cite{zheng2002communication}. Furthermore, it is observed by Han and Rosenthal \cite{han2006geometrical} that maximizing the chordal distance is the appropriate design criterion for the design of unitary space-time constellation at the low-SNR regime in non-coherent wireless communication systems. 
 In another related work by Hochwald \emph{et al.} , a lower bound on pairwise error probability between two subspaces of equal dimension is derived in \cite{hochwald2000systematic}, which is then used to derive a certain criterion for constellation design by maximizing the \textit{chordal} distance between subspaces in the Grassmann space. They also provide a Fourier-based construction method and its equivalent algebraic construction employing linear codes to design constellations for non-coherent MIMO channels. Their approach involves random search procedure and, hence, is not scalable with $n$. Another numerical optimization method is described in \cite{agrawal2001multiple}. Other constructions of Grassmann packings based on Nordstrom-Robinson Codes \cite{aggarwal2006grassmannian} and Reed-Muller codes \cite{ashikhmin2010grassmannian} were also shown to closely approximate the channel capacity of non-coherent MIMO channels.

Another major related line of work includes a wide range of signal processing tasks where the pairwise geometry of subspaces plays an important role in characterizing the performance of the system, e.g., the misclassification probability for the optimal maximum-a-posteriori (MAP) classifier in a Gaussian mixture model (GMM). In particular, a duality  between the problems of classification of $k$-dimensional subspaces from noisy features and the communication over non-coherent MIMO channel are observed in \cite{reboredo2014compressive} and \cite{nokleby2015discrimination}. Also, the results on the capacity of non-coherent MIMO channels are used to provide necessary conditions for successful classification in \cite{nokleby2013information}.  The probability of misclassification is further analyzed in \cite{huang2015role} and is characterized in terms of the principal angles.

Note that non-coherent MIMO communication can be considered as a special case of the non-coherent wireless networking scenario, described by \eq{cheqs}, where there is no interference term while discarding additive noise. Also, there is no relay node and the number of transmitter and receiver antennas are known. In other words, the structure of the network and consequently, the rank of transform matrix $\bH$ is known. Accordingly, prior works on non-coherent MIMO do not deal with subspace errors and erasures, and the underlying communication channel model is totally different from the analog operator channel considered in this paper, as defined in Definition\,\ref{operatorchannel}. 

Note also that the aforementioned prior works on code constructions in the Grassmann space do not often provide \textit{good} solutions, in terms of the trade-off between $R$ and the normalized minimum distance, for general $m$ and $n$. However, a comparison, in terms of the trade-off between $R$ and the normalized minimum distance, between a new construction of analog subspace codes proposed in this paper and the best existing constructions in the literature is done in Section\,\ref{sec:cp}.

\section{Analog Metric Space, Subspace Codes, \\and Error Correction}
\label{sec:three}

In this section we provide a precise description of the \textit{chordal} distance defined for Grassmann space. Then we extend and modify the \textit{chordal} distance to arrive at a new notion of distance, defined for the set of all subspaces of the ambient space, i.e., $\cP(W)$, and show that it conveniently captures the error-correction capability of subspace codes when used over analog operator channels. 

Given the structure of analog operator channels and their input and output alphabets being $\cP(W)$, as defined in Definition\,\ref{operatorchannel}, it is natural to design codes over $\cP(W)$ in order to correct errors and erasures associated with such channels. To this end, the first step is to define a distance function that \textit{properly} captures the effect of errors and erasures imposed by the analog operator channel. Before that, we discuss the \emph{chordal} distance between two $m$-planes, which makes $G_{m,n}(\Lf)$ a metric space. 

The chordal distance $d_c : G_{m,n}(\Lf) \times G_{m,n}(\Lf) \rightarrow \R $ was first introduced for $\Lf=\R$ in \cite{conway1996packing} and was extended to $\Lf = \C$ in \cite{barg2002bounds}. Consider two $m$-planes $U$ and $V$. Let $\bu_i \in U$ and $\bv_i \in V$ be row vectors having unit length such that $|\bu_i \bv_i^{\text{H}}|$ is maximal, subject to the conditions $\bu_i \bu_j^{\text{H}}=0$ and $\bv_i \bv_j^{\text{H}}=0$ for all $i,j$ with $i > j \geq 1$. Then the \textit{principal} angle $\theta_i$, for $i \in [m]$, between $U$ and $V$ is defined as $\theta_i=\arccos{|\bu_i \bv_i^{\text{H}}|}$, see, e.g., \cite{barg2002bounds,roy1947note}. Then the chordal distance between $U$ and $V$ is defined as follows:
\be
\label{chordal distance deff}
d_c(U,V)\,\deff\, \sqrt{\sum_{i=1}^m \sin^2(\theta_i)}.
\ee
Note that this is not the only possible definition for distance between subspaces (see, e.g., \cite{barg2002bounds,ye2014distance} for other similar notions). However, we focus on extending a certain variation of this notion of distance in this paper, to be discussed next. We will show later that it can be made suitable for capturing the error correcting capabilities of subspace codes for the analog operator channel. 

Let $\bZ$ denote an orthonormal matrix spanning $V \in  G_{m,n}(\Lf)$, i.e. ,
$$
V=\Span{\bZ}, \quad \bZ \bZ^{\text{H}} = \bI_m.
$$
Then, the matrix $\bP_V=\bZ^{\text{H}}\bZ$ is an orthogonal projection operator from $\Lf^n$ on $V$. It is shown in \cite{conway1996packing} that $G_{m,n}(\R)$ with chordal distance can be isometrically embedded into a sphere in the Euclidean space $\R^D$ , where $D= {{n+1}\choose{2}}-1$, using the projection matrices associated with subspaces. More specifically,
\be
\label{smallsphere}
\norm{\bP_V-\frac{m}{n}\bI_n}^2=\sqrt{\frac{m(n-m)}{n}},
\ee
for all $V \in G_{m,n}(\R)$. It is also shown that the chordal distance between two $m$-planes is equal to
\be
\label{generalized chordal distance}
d_c(U,V)=\frac{1}{\sqrt{2}}\norm{\bP_U-\bP_V}.
\ee
Moreover, $G_{0,n}(\R),G_{1,n}(\R),\hdots,G_{n,n}(\R)$ can also be embedded into a larger sphere in $\R^{D+1}$, i.e., for all subspaces $V$ of the ambient space $\R^n$ we have
\be\label{large sphere}
\norm{\bP_V-\frac{1}{2}\bI_n}^2=\frac{1}{4}n.
\ee
Note that $G_{m,n}(\R)$ is the intersection of this sphere with the plane described by $\text{tr}(\bP_V)=m$, which is characterized by \eqref{smallsphere}  (see \cite[Figure 10]{ conway1996packing}). Since the Frobenius norm induces a  metric on the set of all $n\times n$ matrices, regardless of whether they are projection matrices or not, one can use \eqref{generalized chordal distance} to generalize the notion of chordal distance to subspaces of different dimensions. This generalized distance is the Euclidean norm of the \emph{chord} connecting the  points associated with $U$ and $V$ on the  sphere characterized by \eqref{large sphere} normalized by $\sqrt{2}$. Note that, this definition coincides with \eqref{chordal distance deff} if $U$ and $V$ have equal dimensions. The proposed generalization of the  chordal distance definition that  also includes subspaces with different dimensions is similar to the  one considered in \cite[Definition 1]{ashikhmin2010grassmannian}. The only minor difference is that the one considered in \cite[Definition 1]{ashikhmin2010grassmannian} has an extra multiplicative factor of $\sqrt{2}$.

Note also that principal angles, used in the definition of chordal distance, do not depend on the choice of basis for the ambient space. In other words, roughly speaking, the chordal distance is invariant under rotation of subspaces. This is shown more formally in the following lemma. 


\begin{lemma}
\label{rotation invariant}
Let $U,V \in \cP(W)$. Given any two orthonormal bases for $W$, namely $\{\boldsymbol{e}_1,\hdots,\boldsymbol{e}_n\}$ and $\{\boldsymbol{e}'_1 ,\hdots,\boldsymbol{e}'_n\}$, referred to as basis $1$ and basis $2$, respectively, we have  
$$
\norm{\bP_U-\bP_V}=\norm{\bP'_U-\bP'_V},
$$
where $\bP_{V}$ and $\bP'_{V}$ are matrix representations for the orthogonal projection operator on $V$ in basis $1$ and basis $2$, respectively.
\end{lemma}
\begin{proof}
Let  $\bZ$ and $\bT$ be  matrices with orthonormal rows, represented in basis $1$, which  span $U$ and $V$, respectively. Then, they are represented in basis $2$ as follows:
$$
\bZ'= \bZ \bQ, \quad \bT'= \bT \bQ,
$$
for a unitary matrix $\bQ$. Then we have the following series of equalities by noting that $\text{tr}(\bA \bB)=\text{tr}(\bB\bA)$ for any two matrices $\bA$ and $\bB$ such that both $\bB\bA$ and $\bA\bB$ are well-defined and that the projection matrices are Hermitian: 
\begin{align*}
   &\norm{\bP'_U-\bP'_V}=\norm{\bQ^{\text{H}} (\bP_{V}-\bP_{U} ) \bQ}\\
   &=\sqrt{\text{tr}(\bQ^{\text{H}} (\bP_{V}-\bP_{U}) \bQ \bQ^{\text{H}} (\bP_{V}-\bP_{U}) \bQ )}\\
   &=\sqrt{\text{tr}(\bQ^{\text{H}}(\bP_{V}-\bP_{U})^2\bQ)}=\sqrt{\text{tr}(\bQ\bQ^{\text{H}}(\bP_{V}-\bP_{U})^2)}\\
   &=\sqrt{\text{tr}((\bP_{V}-\bP_{U})^2)}=\norm{\bP_U-\bP_V},
\end{align*}
which complete the proof. 
\end{proof}

The generalized chordal distance, discussed above, is further modified to arrive at a new notion of \textit{distance} over $\cP(W)$, defined as follows. 
\begin{definition}
\label{dist-def}
The \textit{distance} $d:\cP(W)\times \cP(W)\xrightarrow{} \R$ is defined as
\be
\label{d function}
d(U,V)\,\deff\,\norm{\bP_U-\bP_V}^2=\text{tr}\bigl((\bP_U-\bP_V)^2\bigr),
\ee
where $U,V \in \cP(W)$ and $\bP_U, \bP_V$ are the projection matrices associated to $U,V$, respectively. 
\end{definition}

Note that \Lref{rotation invariant} implies that $d(.,.)$ is well-defined. Note also that $d(.,.) = 2d_c(.,.)^2$ by \eqref{generalized chordal distance} and \eqref{d function}  and, equivalently, is equal to the square of the metric considered in \cite[Definition\,1]{ashikhmin2010grassmannian}. 
It is shown in \Lref{dist-quasidist-relation} in the appendix that the square of a metric is a $2$-quasimetric, where a quasimetric is defined in Section\,\ref{notation}. Hence, $d(.,.)$ is a $2$-quasimetric. It is further shown in \Lref{sim_diag} in the appendix that for $U,V,T \in \cP(W)$, $d(.,.)$ satisfies the triangle inequality, i.e., $\sigma = 1$ in \eq{quasi-triangle-ineq}, as long as $P_U$ and $P_V$ are simultaneously diagonalizable, i.e., one can find a basis in which both $P_U$ and $P_V$ are diagonal matrices. This property is later utilized to characterize the error-and-erasure correction capability of codes used over analog operator channels in terms of their \textit{minimum distance} the same way it is done given an underlying  metric. Hence, we refer to $d(.,.)$ as a distance through the rest of this paper keeping in mind that it is indeed a $2$-quasimetric.

An equivalent expression for the distance $d(.,.)$ is derived in the following lemma.

\begin{lemma}
\label{another expression for chordal distance}
For any $\Span{\bZ},\Span{\bT} \in G_{m,n}(\Lf)$, where the rows of $\bZ$ and $\bT$ are orthonormal, we have
$$
d(\Span{\bZ},\Span{\bT})=2(m-\norm{\bZ \bT^{\text{H}}}^2).
$$
\end{lemma}
\begin{proof}
By noting that $\text{tr}(\bZ\bZ^{\text{H}})=\text{tr}(\bT\bT^{\text{H}})=m$ and by using \eqref{d function} one can write
\begin{align*}
   &d(\Span{\bZ},\Span{\bT})=\text{tr}((\bZ^{\text{H}}\bZ-\bT^{\text{H}}\bT)^2)=2(m-\text{tr}(\bZ^{\text{H}}\bZ\bT^{\text{H}}\bT))\\
   &=2(m-\text{tr}(\bT\bZ^{\text{H}}\bZ\bT^{\text{H}}))=2(m-\text{tr}((\bZ \bT^{\text{H}})^{\text{H}}(\bZ \bT^{\text{H}})))\\
   &=2(m-\norm{\bZ\bT^{\text{H}}}^2),
\end{align*}
which completes the proof.
\end{proof}


\begin{definition}
\label{codedef}
An analog subspace code $\cC$ is a subset of $\cP(W)$. The size of $\cC$ is denoted by $|\cC|$. The minimum distance of $\cC$ is defined as
$$
\dmin(\cC)\,\deff\, \min_{U,V\in \cC, U\neq V} d(U,V),
$$
where $d(.,.)$ is defined in Definition\,\ref{dist-def}. The maximum dimension of the codewords of $\cC$ is denoted by
$$
l(\cC)\,\deff\,\max_{U\in \cC} \dim(U).
$$
The code $\cC$ is then referred to as an $[n,l(\cC),|\cC|,\dmin(\cC)]$ subspace code, where $n$ is the dimension of the ambient space $W$.
\end{definition}

If the dimension of all codewords in $\cC$ are equal, then the code is referred to as a \emph{constant-dimension} code, which is also called a code on Grassmannian or a Grassmann code in the literature. 

The \emph{dual} subspace code associated with subspace code $\cC$ is the code $\cC^{\perp}\,\deff\,\{U^{\perp}:U\in \cC \}$. \Lref{orthogonal subspaces distance mirroring} implies that $\dmin(\cC^\perp)=\dmin(\cC)$. Note that if $\cC$ is a constant-dimension code of type $[n,l,M,\dmin]$, then $\cC^\perp$ is a constant-dimension code of type $[n,n-l,M,\dmin]$.

\begin{definition}
\label{code parameters}
Let $\cC$ be an $[n,l,M,\dmin(\cC)]$ subspace code. The normalized weight $\lambda$, the rate $R$, and the normalized minimum distance $\delta$ of $\cC$ are defined as follows:
$$
\lambda\,\deff\,\frac{l}{n}, \quad R\,\deff\,\frac{\ln{M}}{n}, \quad \delta\,\deff\, \frac{\dmin(\cC)}{2l}.
$$
\end{definition}

Note that the normalized weight $\lambda$ and the normalized minimum distance $\delta$ are always between $0$ and $1$. However, while designing constant-dimension codes one can limit the attention to $\lambda \in [0,\frac{1}{2}]$. This is because for any code $\cC$ with $l>\frac{n}{2}$, there exists a dual code $\cC^\perp$ with $l<\frac{n}{2}$ and having the same distance properties.

As in conventional block codes, one can associate a minimum distance decoder to a subspace code $\cC$, e.g., when used for communication over an analog operator channel, in order to recover from subspace errors and erasures. Such a decoder returns the nearest codeword $V \in \cC$ given $U \in \cP(W)$ as its input, i.e., for any $V'\in \cC$, $d(U,V)\leq d(U,V')$. The following lemma plays a key rule in relating the minimum distance of $\cC$ to its error-and-erasure correction capability under minimum distance decoding.

\begin{lemma}
\label{inequality for operator channel}
    Let $U,V \in\cP(W)$ denote the input and the output of an analog operator channel, respectively, with the relation specified in \eqref{channelIOrelation}. Then for any $T\in \cP(W)$ we have
    \be
    \label{tri-ineq-for-opch}
    d(U,T)\leq \rho+t+d(V,T),
    \ee
    where $\rho$ and $t$ denote the dimension of erasures and errors, respectively, as specified in Definition\,\ref{operatorchannel}.
\end{lemma}
\begin{proof}
    Let $U'=\cH_k(U)$. Then $d(U,U')=\rho$ by \Lref{direct sum related to distance}. Also, as shown in the proof of \Lref{direct sum related to distance}, $P_U$ and $P_{U'}$ are simultaneously diagonalizable. Hence, by \Lref{sim_diag} we have
    \be\label{triangle inequality for erasure}
    d(U,T)\leq \rho+d(U',T),
    \ee
    for any $T\in \cP(W)$. Moreover, $V=U'\oplus E$, where $E$ denotes the error space with $\dim(E)=t$. Using the same argument $P_{U'}$ and $P_V$ are also simultaneously diagonalizable. Similarly, by \Lref{sim_diag} we have  
    \be\label{triangle inequality for error}
    d(U',T)\leq t+d(V,T),
    \ee
    for any $T\in \cP(W)$. Combining \eqref{triangle inequality for erasure} with \eqref{triangle inequality for error} completes the proof.
\end{proof}

\begin{theorem}
\label{error correction capability of codes related to minimum distance}
Consider a subspace code $\cC$ used for communication over an analog operator channel, as defined in Definition\,\ref{operatorchannel}, i.e., the input to the channel is $U \in \cC$. Let $t$ and $\rho$ denote the dimension of errors and erasures, respectively. Then the minimum distance decoder successfully recovers the transmitted codeword $U \in \cC$ from the received subspace $V$ if
\be
\label{correct decoding condition}
2(\rho+t)<\dmin(\cC),
\ee
where $\dmin(\cC)$ is the minimum distance of $\cC$ defined in Definition\,\ref{codedef}. 
\end{theorem}
\begin{proof}
By \Lref{inequality for operator channel} we have
\be\label{inequality IO}
d(U,T)\leq \rho+t+d(V,T),
\ee
for any $T\in \cP(W)$. In particular,
\be \label{distance UV}
d(U,V)\leq \rho+t,
\ee
by letting $T=V$. Now, let $T \in \cC$ be a codeword other than $U$. By \eqref{correct decoding condition} and the definition of minimum distance we have
\be
d(U,T) \geq \dmin(\cC) > 2(\rho+t).
\ee
This together with \eqref{inequality IO} and \eqref{distance UV} yields
\be\label{distance VT}
d(V,T) > \rho+t \geq d(U,V).
\ee
Hence, the minimum-distance decoder returns $U$.
\end{proof}

\Tref{error correction capability of codes related to minimum distance} implies that erasures and errors have equal costs in the subspace domain as far as the minimum-distance decoder is concerned. In other words, the minimum-distance decoder for a code $\cC$ can correct up to $\floor{\frac{\dmin(\cC)-1}{2}}$ errors and erasures.

\noindent
{\bf Remark\,2.} If one uses the chordal distance, instead of the distance $d(.,.)$ defined in Definition\,\ref{dist-def}, and follows the similar arguments as we did in this section, a result similar to \Tref{error correction capability of codes related to minimum distance} can be obtained while the condition in \eqref{correct decoding condition} is replaced by $\sqrt{2}(\sqrt{\rho} + \sqrt{t})$ being strictly less than the minimum chordal distance of the code. Since $d(.,.) = 2d_c(.,.)^2$, where $d_c(.,.)$ is the generalized chordal distance, this condition can be expressed in terms of $\dmin(\cC)$ as follows:
\be
\label{new_dist_cond}
4(\sqrt{\rho} + \sqrt{t})^2<\dmin(\cC).
\ee
Note that the left hand side of \eqref{new_dist_cond} is greater than that of \eqref{correct decoding condition} by a multiplicative factor that is between $2$ and $4$. This shows the clear advantage in using the new distance $d(.,.)$ instead of the chordal distance in characterizing the error-and-erasure correction capability of analog subspace codes. The advantage is due to the fact that although $d(.,.)$ does not always satisfy the triangle inequality, it exhibits properties of a  metric when dealing with inputs and outputs of analog operator channels.


\section{Robustness Against Additive Noise}
\label{sec:four}

In this section, we analyze the robustness of the analog operator channel in the presence of an additive noise.

The additive noise is denoted by $\bN\in \C^{l\times n}$, referred to as the \emph{noise matrix}. In the presence of the additive noise, the transform equation described in \eqref{cheqs} is extended as follows:
\be
\label{matrixchannelwithnoise}
\bY_{l \times n}=\bH_{l\times m}\bX_{m \times n}+\bG_{l \times t}\bE_{t \times n} + \bN_{l \times n}.
\ee
More specifically, the effect of all the noise terms added to the blocks across the wireless network is included in the noise matrix $\bN$. For ease of notation, let $\bA$ denote the term $\bH\bX+\bG\bE$, consisting of terms associated to the transmitted blocks as well as the interference blocks, referred to as the \emph{signal matrix}. 

In the rest of this section, we aim at characterizing the perturbation imposed by the additive noise in terms of the subspace distance. In other words, we derive bounds on subspace distance between the signal matrix and the signal matrix perturbed by noise, i.e., $d(\Span{\bA},\Span{\bA+\bN})$, in terms of various characteristics of both the noise and the signal matrix.  

First, we consider the case where $\bA$ is full row rank. In a sense, this implies that the receiver does not \textit{oversample} from the network. In this case, we show that the subspace distance caused by the noise is bounded in terms of two parameters: (1) the ratio of the maximum singular value, also referred to as the spectral norm, of the signal matrix to the Frobenius norm of the noise matrix, i.e.,
$\frac{\norm{\bA}_2}{\norm{\bN}}$; and (2) the spectral norm condition number of the signal matrix, defined as 
\be\label{condnumdeff}
\kappa(\bA)\,\deff\,\norm{\bA}_2|\hspace{-.2mm}|{\bA^+}|\hspace{-.2mm}|_2,
\ee
where $\bA^+$ is the pseudo-inverse of $\bA$, defined as $\bA^+\,\deff\, \bA^{\text{H}}(\bA\bA^{\text{H}})^{-1}$. It is well-known that $\kappa(\bA)=\frac{\sigma_{\max}}{\sigma_{min}}\geq 1$, where $\sigma_{\max}=\norm{\bA}_2$ and $\sigma_{\min}$ are the largest and the smallest singular values of $\bA$, respectively.


Our analysis is based on the analysis of perturbation of RQ-factorization, see., e.g., \cite{sun1991perturbation}. Recall, from the linear algebra literature, that the RQ-factorization of a given matrix $\bA \in \C^{l \times n}$ decomposes it as $\bA=\bR\bQ$, where $\bQ\in \C^{l\times n}$ is an orthonormal matrix with $\bQ\bQ^{\text{H}}=\bI_l$ and $\bR\in \C^{l\times l}$ is an upper triangular matrix with positive diagonal elements. It is well known that the RQ-factorization is unique \cite{golub1996matrix }. The following result relates the perturbation of $\bA$ to the perturbation of  $\bQ$ in its RQ-factorization. 

\begin{theorem}\cite[Theorem 1.6]{sun1991perturbation} (rephrased)
\label{qbound}
    Let  $\bA \in \C^{l \times n}$ be a full row rank matrix with RQ-factorization $\bA=\bR\bQ$. Then for any $\bE\in \C^{l\times n}$ that satisfies $\norm{\bA^+}_2\norm{\bE}_2<1$, we have the following RQ-factorization for $\bA+\bE$:
    $$
    \bA+\bE=(\bR+ \bR_{\Delta})(\bQ+\bQ_{\Delta}),
    $$ 
    with $(\bQ+\bQ_{\Delta})(\bQ+ \bQ_{\Delta})^{\text{H}}=\bI_l$ such that
    \begin{align}
        \norm{\bQ_{\Delta}}\leq \frac{(1+\sqrt{2})\kappa(\bA)}{1-\norm{\bA^+}_2\norm{\bE}_2}.\frac{\norm{\bE}}{\norm{\bA}_2}.
        \label{perturbation bound}
    \end{align}
    \end{theorem}

An upper bound on the subspace distance caused by the noise is derived in the following theorem.

\begin{theorem}\label{perturbation bound}
 Let  $\bA \in \C^{l \times n}$ be a full row rank matrix. Then for any $\bE\in \C^{l\times n}$ that satisfies $\norm{\bA^+}_2\norm{\bE}_2<1$, we have
    \be\label{subspace perturbation}
    d(\Span{\bA},\Span{\bA+\bE})\leq 2\epsilon+\epsilon^2,
    \ee
    where
    \be\label{epsilon bound}
    \epsilon\,\deff\, ( \frac{(1+\sqrt{2})\kappa(\bA)}{1-\norm{\bA^+}_2\norm{\bE}_2}.\frac{\norm{\bE}}{\norm{\bA}_2})^2.
    \ee
\end{theorem} 

\begin{proof}
Suppose that $\bQ$ and $\bQ_{\Delta}$ are derived according to \Tref{qbound}. Then we have
\begin{align}
    d(\Span{\bA},\Span{\bA+\bE})&=\norm{ (\bQ+\bQ_{\Delta})^{\text{H}}(\bQ+\bQ_{\Delta})-\bQ^{\text{H}}\bQ}^2\label{deff}\\
    &\leq \norm{\bQ^{\text{H}}\bQ_{\Delta}}^2+\norm{\bQ_{\Delta}^{\text{H}} \bQ}^2+\norm{\bQ_{\Delta}^{\text{H}} \bQ_{\Delta}}^2\label{triangle}\\
    &=2\norm{\bQ^{\text{H}}\bQ_{\Delta}}^2+\norm{\bQ_{\Delta}^{\text{H}} \bQ_{\Delta}}^2,\label{hermitian}
\end{align}
where \eqref{deff} is by \eqref{d function}, \eqref{triangle} follows from triangle inequality and \eqref{hermitian} is by observing that $\norm{\bA}=\norm{\bA^{\text{H}}}$ for any matrix $\bA$. Moreover, for the first term in \eqref{hermitian} we have
\begin{align}
 \norm{\bQ^{\text{H}} \bQ_{\Delta}}^2=\text{tr}(\bQ_{\Delta}^{\text{H}}\bQ\bQ^{\text{H}} \bQ_{\Delta})=\text{tr}(\bQ_{\Delta}^{\text{H}} \bQ_{\Delta})=\norm{\bQ_{\Delta}}^2,\label{first term}
 \end{align}
since $\bQ\bQ^{\text{H}}=I_l$. Applying \Lref{BHB}, in the appendix, to the second term in  \eqref{hermitian} results in 
\be
\label{second term}
\norm{\bQ_{\Delta}^{\text{H}} \bQ_{\Delta}}\leq \norm{\bQ_{\Delta}}^2.
\ee
Then \eqref{hermitian}, \eqref{first term}, and \eqref{second term} together with the result of \Tref{qbound} yield \eq{subspace perturbation}.
\end{proof}

The result of \Tref{perturbation bound} can be applied to upper bound the subspace distance $d(\Span{\bA},\Span{\bA+\bN})$ caused by the additive noise $\bN$ as long as the signal matrix $A$ is full row rank. If $A$ is not full row rank, which means that the receiver is somewhat \textit{oversampling} from the network, the addition of $\bN$, even with very small norm, may result in a large $d(\Span{\bA},\Span{\bA+\bN})$. This is because $\bA+\bN$ is full row rank with probability $1$ if entries of $N$ are Gaussian random variables. Consequently, by \Lref{direct sum related to distance},
\be
\label{noiserank}
d(\Span{\bA},\Span{\bA+\bN}) \geq l - \rank(\bA),
\ee
regardless of how small $\norm{\bN}$ is. The aim here is to characterize the error correction capability of subspace codes in the presence of additive noise, i.e., to extend the result of \Tref{error correction capability of codes related to minimum distance} to take into account the effect of the additive noise as well as subspace errors and erasures. In order to do so for the general case, where the signal matrix $\bA$ is not necessarily full row rank, one can model the effect of the noise partially as an interference of dimension $l - \rank(\bA)$ and partially as an addition of noise to a full row rank sub-matrix of the signal matrix, whose effect can be bounded using \Tref{perturbation bound}. This is elaborated through the rest of this section.

Let $r=\rank(\bA)$ and $r_d$ denote the rank-deficiency of $\bA$, i.e.,
\be
\label{rd-def}
r_d\,\deff\,l - r. 
\ee
Also, let $\bA_1$ be an $r\times n$ full row rank sub-matrix of $\bA$ and $\bA_2$ denote the sub-matrix of $\bA$ consisting of its remaining rows. Let $\bN_1$ and $\bN_2$ be sub-matrices of $\bN$ with row indices associated with row indices of $\bA_1$ and $\bA_2$, respectively. Without loss of generality one can write 
\be\label{decomposition}
 A=
\left[
\begin{array}{c}
\bA_1\\
\hline
\bA_2
\end{array}
\right],  \quad  N=
\left[
\begin{array}{c}
\bN_1\\
\hline
\bN_2
\end{array}
\right],
\ee
where both $\bA_1$ and $\bN_1$ are $r \times n$ matrices and both $\bA_2$ and $\bN_2$ are $r_d \times n$ matrices. Then we have the following theorem.

\begin{theorem}\label{perturbation bound for the general case}
    Let $\bA \in \C^{l\times n}$ with $\rank(\bA)=r$. Let $\bA_1 \in \C^{r\times n}$ denote a foll row rank sub-matrix of $\bA$. Then for any $\bN \in \C^{l \times n}$ that satisfies $\norm{\bA_1^+}_2\norm{\bN}_2<1$ we have
\be
\label{bound for the first term}
d(\Span{\bA}, \Span{\bA+\bN})\leq (\sqrt{r_d}+\sqrt{\Delta})^2,
\ee
where $r_d$ is the rank-deficiency of $\bA$, as defined in \eq{rd-def}. Also, 
\be
\label{Delta-def}
\Delta\,\deff\, 2\epsilon+\epsilon^2,
\ee
where
\be \label{epsilon'deff}
\epsilon\,\deff\,( \frac{(1+\sqrt{2})\kappa(\bA_1)}{1-\norm{\bA_1^+}_2\norm{\bN}_2}.\frac{\norm{\bN}}{\norm{\bA_1}_2})^2.
\ee
\end{theorem} 
\vspace{1mm}
\begin{proof} Let $\bA_1$, $\bA_2$, $\bN_1$, and $\bN_2$ be as specified in \eqref{decomposition}. Then we have the following:
  \begin{align}
      &d(\Span{\bA},\Span{\bA+\bN})^{\shalf}\\
      &=d(\Span{\bA_1},\Span{\bA+\bN})^{\shalf}\label{subspace equivalence}\\
      &\leq d(\Span{\bA_1}, \Span{\bA_1+\bN_1})^{\shalf}+ d(\Span{\bA_1+\bN_1},\Span{\bA+\bN})^{\shalf},    \label{triangle inequality }
  \end{align}
where \eqref{subspace equivalence} holds because $\Span{\bA}=\Span{\bA_1}$ and \eqref{triangle inequality } is by triangle inequality  for $d_c(.,.)$ and noting that $d(U,V)=2d_c(U,V)^2$. We will bound the two terms in \eqref{triangle inequality } separately. 

To bound the first term in \eqref{triangle inequality }, first note that $\norm{\bN_1}\leq \norm{\bN}$. Also, for the spectral norm, we have $\norm{\bN_1}_2\leq \norm{\bN}_2$ since adding a row to a rectangular matrix does not reduce its maximum singular value, see, e.g., \cite{hwang2004cauchy}. This together with \Tref{perturbation bound} yield the following upper bound on the first term in \eqref{triangle inequality }:
\be
\label{bound for the first term}
d(\Span{\bA_1}, \Span{\bA_1+\bN_1})\leq \Delta = 2\epsilon+\epsilon^2
\ee
where $\epsilon$ is specified in \eq{epsilon'deff}. 

For the second term in \eqref{triangle inequality } we have 
\begin{align}
\nonumber
     &d(\Span{\bA_1+\bN_1},\Span{\bA+\bN})\\
     &= d(\Span{\bA_1+\bN_1},\Span{\bA_1+\bN_1}+\Span{\bA_2+\bN_2})\label{sum of spaces}\\
     &\leq \rank(\Span{\bA_2+\bN_2})\leq r_d, \label{directsum }
\end{align}
where \eqref{sum of spaces} holds by the definition in \eq{plus-def}, and \eqref{directsum } follows from \Lref{direct sum related to distance} and by noting that for any $U,V \in \cP(W)$ one can always write $U+V=U\oplus V'$ for some $V'\in \cP(W)$ with $\dim(V')\leq \dim(V)$. The proof is complete by combining \eqref{bound for the first term} and \eqref{directsum } together with \eqref{triangle inequality }.
\end{proof}

\Tref{perturbation bound for the general case} shows that the additive noise affects the output of the analog operator channel in two ways. It, in a sense, \textit{rotates} the output subspace by a value upper bounded by $\Delta$ and also, implicitly, induces an extra interference term of dimension upper bounded by $r_d$ (For simplicity, we consider the worst case scenario where this dimension is $r_d$). This motivates us to define the \textit{noisy} analog operator channel as follows. First, we define a stochastic operator $\cR_\Delta$, referred to as the \emph{rotation operator}, which takes a subspace $U\in \cP(W)$ as the input and returns a random subspace $V \in \cP(W)$ with $\dim(V)=\dim(U)$ as the output in such a way that 
$$
d(U,V) \leq \Delta.
$$ 
\begin{definition}\label{noisy operator channel}
   A noisy analog operator channel associated with the analog ambient space $W$ is a channel with input $U$ and output $V$, where $U,V \in \cP(W)$, with the following input-output relation:
   \be
   \label{noisychannelIOrelation}
   V=\cR_\Delta\bigl(\cH_k(U)\oplus E\bigr)\oplus F,
   \ee
where $\cH_k(U)$ and $E$ induce subspace erasures and errors, respectively, as in the analog operator channel model, defined in \Dref{operatorchannel}, $\cR_\Delta$ is the rotation operation defined above, and $F$ is the implicit interference caused by the additive noise. 
\end{definition}

The following theorem extends the result of \Tref{error correction capability of codes related to minimum distance} to take into account the effect of the additive noise as well as the subspace errors and erasures. 

\begin{theorem}\label{minimum distance decoder in presence of the additive noise}
Consider a subspace code $\cC$ used for communication over a noisy analog operator channel, as defined in Definition\,\ref{noisy operator channel}, i.e., the input to the channel is $U \in \cC$. Let $t$, $\rho$, and $r_d$ denote the dimension of errors, erasures, and the implicit noise interference $F$, respectively. Then the minimum distance decoder successfully recovers the transmitted codeword $U \in \cC$ from the received subspace $V$ if
\be
\label{radius_with_noise}
\rho+t + (\sqrt{\rho+t+\Delta} + \sqrt{\Delta} + 2 \sqrt{r_d})^2 < \dmin(\cC).
\ee
\end{theorem}

\begin{proof}
Let $V_1 = \cH_k(U)\oplus E$ and $V_2 = \cR_\Delta\bigl(V_1)$. Note that we have $V = V_2 \oplus F$. Then by applying \Lref{inequality for operator channel} to the analog operator channel with input $U$ and output $V_1$ we have
\be
\label{eq11}
d(U,V_2)\leq \rho+t+d(V_1,V_2) = \rho+t + \Delta.
\ee
Also, by using the triangle inequality for the chordal distance $d_c(.,.)$ and noting that $d(.,.) = 2d_c(.,.)^2$ we have
\be
\label{eq12}
d(U,V) \leq \bigl( \sqrt{d(U,V_2)} + \sqrt{d(V_2,V)}\bigr)^2.
\ee
By \Lref{direct sum related to distance} we have $d(V_2,V) = r_d$. This together with \eqref{eq11} and \eqref{eq12} yields
\be
\label{eq13}
d(U,V) \leq (\sqrt{\rho+t + \Delta} + \sqrt{r_d})^2.
\ee
Now consider a codeword $T \in \cC$ other than $U$. Again, by applying \Lref{inequality for operator channel} to the analog operator channel with the input $U$ and the output $V_1$ and by rearranging the terms we have
\be
\label{eq14}
d(T,V_1) \geq d(T,U) - \rho - t \geq \dmin(\cC) - \rho - t,
\ee
where the last inequality is by the definition of minimum distance. Also, by using the triangle inequality for the chordal distance we have
\be
\label{eq15}
d(T,V) \geq \bigl(\sqrt{d(T,V_1)}- \sqrt{d(V_1,V_2)} - \sqrt{d(V_2,V)}\bigr)^2.
\ee
Again by noting that $d(V_1,V_2) \leq \Delta$, $d(V_2,V) = r_d$, and by combining \eqref{eq15} with \eqref{eq14} we have
\be
\label{eq16}
d(T,V) \geq \bigl(\sqrt{\dmin(\cC) - \rho - t}- \sqrt{\Delta} - \sqrt{r_d}\bigr)^2.
\ee
It can be observed that the condition given in \eqref{radius_with_noise} implies that the right hand side of \eqref{eq16} is strictly greater than that of \eqref{eq13}. In other words, \eqref{radius_with_noise} guarantees that
$$
d(T,V) > d(U,V),
$$
for any codeword $T \in \cC$ other than $U$. Hence, the minimum distance decoder returns $U$ which completes the proof. 
\end{proof}
\noindent
{\bf Remark\,3.} Note that \Tref{minimum distance decoder in presence of the additive noise} reduces to \Tref{error correction capability of codes related to minimum distance} by setting $r_d = \Delta = 0$. In other words, \Tref{minimum distance decoder in presence of the additive noise} \textit{properly} extends the result of \Tref{error correction capability of codes related to minimum distance}, on relating the minimum distance of analog subspace codes to their error-and-erasure correction capability, to the noisy analog operator channel scenario. In practice, the implicit noise interference term $F$ and, consequently, the term $r_d$ in \eq{radius_with_noise} can be potentially removed by simply discarding a certain number of received blocks at the receiver. However, this requires knowing the rank of the received signal by the receiver which may not be readily available due to the assumptions on non-coherent communications. This can be further explored when considering a practical wireless networking scenario to see whether such information, i.e., the rank of the received signal, can be obtained or well-approximated, e.g., using principal component analysis (PCA) methods, by the receiver. Also, as shown in \Tref{perturbation bound for the general case}, the other term, besides $r_d$, resulting from the additive noise that affects the output subspace is $\Delta$. Note that for a fixed signal matrix $\bA$, as $\norm{\bN} \rightarrow 0$, we have $\epsilon \rightarrow 0$ as well as $\Delta \rightarrow 0$, where $\epsilon$ and $\Delta$ are specified in \eq{epsilon'deff} and \eq{Delta-def}, respectively. This together with a procedure to remove the $r_d$ term, as discussed above, show that the analog operator channel can be made \textit{robust} with respect to the additive noise, i.e., the subspace distance perturbation in the received signal matrix, caused by the additive noise, goes to zero as $\norm{\bN} \rightarrow 0$.

\noindent

\section{Constructions of Analog Subspace Codes}
\label{sec:five}

In this section, we explore three approaches for constructing analog subspace codes. In particular, the novel approach based on character sums results in explicit constructions with better rate-minimum distance trade-off compared to the previously known constructions for a wide range of parameters.

More specifically, we are concerned with the following equivalent questions. For a given dimension of the ambient space $n$ and the size of the subspace code $|\cC|$, or equivalently the rate of $\cC$, construct $\cC$ with the maximum possible $\dmin$. Alternatively, given $n$ and a desired $\dmin$, construct the subspace code $\cC$ with the maximum size/rate. Although the exact answer to these equivalent questions are not known in general, one can derive sphere-packing-type upper bounds and Gilbert-Varshamov-type lower bounds for codes in the subspace domain.

By precisely characterizing the volume of balls in the Grassmann space Barg and Nogin derived lower and upper bounds for $R$ as $n\xrightarrow{} \infty$ while $m$ and $\delta$ are fixed \cite{barg2002bounds}. More specifically, they analyzed the asymptotic behavior of the volume of a sphere with a fixed radius on $G_{m,n}(\Lf)$ that is then used in a packing-type and a covering-type argument. Their result is recapped in the following theorem. Note that we use the new notion of distance, as defined in Definition\,\ref{dist-def}, to state the results.

\begin{theorem}\cite[Theorem 2]{barg2002bounds}
\label{barg}
There exists a sequence of codes in $G_{m,n}(\Lf)$ with fixed $m$ and $\delta$, while $n \rightarrow \infty$, and asymptotic rate 
\be \label{barg lower bound}
R > -\frac{1}{2} \beta m \ln(\delta).
\ee
Moreover, for any such sequence of codes 
\be 
\label{barg upper bound}
R < -\beta m \ln(\sqrt{1-\sqrt{1-\frac{\delta}{2}}}),
\ee
where $\beta=1,2$ for $\Lf=\R,\C$, respectively, as discussed in Section\,\ref{notation}.
\end{theorem}

Note that the lower bound in \Tref{barg} is based on existence-type arguments. However, a somewhat \textit{stronger} result can be established to conclude that such codes, perhaps by sacrificing in the rate, can be found with probability arbitrarily close to $1$ in a random ensemble. This is the focus of the next subsection. Also, it is discussed how such a result can be used in constructing explicit codes. 

\subsection{Constructions based on random ensembles}

For large values of $n$, the volume of a ball with a certain radius $r$ in $G_{m,n}(\C)$ is characterized in \cite{barg2002bounds}. This is done by utilizing the relation between the principal angles of two uniformly distributed subspaces in $G_{m,n}(\C)$ and the singular values of Wishart matrices, the matrices of the form $\bN \bN^{\text{H}}$, where the elements of $\bN \in \C^{m \times n}$ are i.i.d standard normal  random variables. Note that the subspace spanned by the rows of such a random matrix $\bN$ is uniformly distributed on the corresponding sphere in the Grassmann space $G_{m,n}(\C)$. 

Let $\cR$ denote a random ensemble of subspace codes with code size $M=\exp(nR)$, wherein each codeword is the row space of a randomly generated $m \times n$ matrix with i.i.d entries from the $\cN(0,1)$ distribution. In the next theorem, it is shown that a random subspace code \textit{almost} achieves half of the Gilbert-Varshamov lower bound, stated in \eqref{barg lower bound}, as $n \rightarrow \infty$, with a probability that approaches $1$ exponentially fast in $n$. 

\begin{theorem}
\label{random_code_thm}
Consider a random ensemble $\cR$ of subspace codes with the rate 
$$
R=-\frac{1}{4}\beta m\ln(\delta)-\epsilon,
$$
for some $\epsilon>0$. Then the normalized minimum distance of a code $\cC$ randomly picked from $\cR$ is at least $\delta$ with probability at least $1-\exp(-2n\epsilon+o(n))$. 
\end{theorem}
\begin{proof}
Let $\cC = \{\Phi_i: 1\leq i \leq M\}$, where $M=\exp(nR)$ and $\cC$ is randomly picked from $\cR$. Then the probability $p$ that two arbitrary codewords $\Phi_i$ and $\Phi_j$ have distance at most $r$ is equal to the volume of a ball with radius $r$, which can be characterized as follows \cite[Theorem 1]{barg2006bound}:
\be \label{ball volume}
p=(\frac{r}{2m})^{\frac{\beta nm}{2}+o(n)}.    
\ee
Let $X_{i,j}$ be an indicator random variable which is $1$ if $d(\Phi_i,\Phi_j)\leq r$; otherwise, $X_{i,j} = 0$. Let 
$$
X\,\deff\, \sum_{i=1}^{M}\sum_{j=i+1}^{M}X_{i,j}.
$$
Using Markov inequality together with \eqref{ball volume} and the linearity of expectation we have:
\begin{align}
  &\Pr[X\geq 1]<E[X]= {M\choose{2}} \Pr[X_{i,j}=1]\\
  &= {M\choose{2}} p<\exp\Bigl(2nR+\bigl(\frac{\beta mn}{2} +o(n)\bigr)\ln(\frac{r}{2m})\Bigr). \label{rate of the random code}
\end{align}
Note that if the random variable $X$, associated to $\cC$, is zero, then it implies that $\dmin(\cC) \geq r$, i.e.,
\be
\label{random_code_thm_eq1}
\Pr[\dmin(\cC) \geq r] = \Pr[X = 0] = 1 - \Pr[X\geq 1].
\ee
This together with \eqref{rate of the random code} complete the proof. 
\end{proof}
\noindent
{\bf Remark\,4.} Note that random ensembles, in general, do not lead to explicit constructions of subspace codes that can be constructed with complexity that is polynomial in $n$. However, a potential approach to utilize such ensembles is to pick another parameter $n'$ that is much smaller than $n$, e.g., $n'=O(\log n)$, and consider a random ensemble of subspace codes with the dimension of ambient space equal to $n'$. Then a brute-force search within the ensemble is feasible as its complexity is exponential in $n'$ and, hence, it is polynomial in $n$. Also, a minimum distance decoding is feasible for such a code. Then a \textit{proper} concatenation scheme can be potentially used, by concatenating the random inner code with some explicit construction of an outer code in order to construct explicit codes with non-vanishing rates given a fixed $\delta$ as $n \rightarrow \infty$. The details are left for future work.

 \subsection{Packing lines using binary codes}

A special and yet practically interesting case of the analog operator channel is when $m=1$. For instance, from the non-coherent wireless communications perspective, elaborated in Section\,\ref{sec:analog_operator}, this can be a reasonable scenario in the uplink transmission of a wireless node with one antenna transmitting the node's data in one time frame. In such cases, it is natural to assume that there is no rank deficiency; otherwise, reliable communication is not possible. A possible approach to construct codes in $G_{1,n}(\Lf)$ is to use well-known constructions of binary linear codes and map the binary coded data into real/complex symbols resembling a joint coded modulation design. 

The idea of constructing codes in $G_{1,n}(\R)$ using binary linear codes  was first suggested in \cite{conway1996packing}. Consider a binary linear code of length $n$ that forms a closed set under the completion, i.e., it contains the all-one codeword, with normalized minimum Hamming distance $\gamma$. Then a possible mapping to real-valued symbols is to map zeros to $1$'s and ones to $-1$'s. This results in a code in $G_{1,n}(\R)$ with the normalized minimum distance $\delta = 32\gamma^2(1-\gamma)^2$. The same result also holds in the complex domain, i.e., packing lines in $G_{1,n}(\C)$, where one can map ones to $(1+i)$'s and zeros to $-(1+i)$'s. Hence, given a family of  binary linear codes with fixed Hamming distance and asymptotically non-vanishing rate in terms of $n$, one can construct a code in $G_{1,n}(\Lf)$ with a fixed minimum distance and non-vanishing rate as well. To this end, in the rest of this subsection, we briefly overview various families of binary linear codes, from this perspective, in the coding theory literature.

There are several well-known constructions of binary linear codes, mainly based on code concatenations, with asymptotically good minimum distance. The idea of code concatenation was first introduced by Elias in the form of product codes \cite{elias} and was later developed by Forney \cite{forney1965concatenated}. Also, the well-known Zyablov lower bound for the normalized minimum distance of a concatenated code with rate $R$ is characterized as follows \cite{zyablov1971estimate}:
\be
\label{zyablov}
\delta_Z(R)=\max_{R\leq x\leq 1} \delta_{GV}(x)(1-\frac{R}{x}),
\ee
where $\delta_{GV}(x)$ is the Gilbert-Varshamov normalized distance at rate $x$ for the binary codes. Furthermore, by using multilevel concatenation, codes with minimum distance even larger than Zyablov bound \eqref{zyablov} can be constructed. More specifically, a generalization to \eqref{zyablov} by letting the number of concatenation levels going to infinity was given later, known as Blokh-Zyablov bound\cite{blokh1982linear}, stated as follows:
\be
\label{Blokh-Zyablov}
R_{BZ}=1-h(\delta)-\delta\int_0^{1-h(\delta)}\frac{dx}{\delta_{GV}(x)}.
\ee

Another line of work on combining codes is due to Tanner \cite{tanner1981recursive} with the general theme of using one or more shorter codes, referred to as subcodes, in combination with a certain bipartite graph. In particular, a certain family of such graph-based codes is referred to as expander codes, which are proved to have asymptotically good minimum distance. For instance, Barg and Zemor \cite{barg2006distance} proposed a family of expander codes meeting the Zyablov bound, specified in \eqref{zyablov}, with the construction complexity at most $O(n^2)$ and a decoder, with complexity that is linear in $n$, that corrects up to half of the minimum distance. They further improved this result by introducing another family of codes that exceeds Zyablov bound with construction complexity not more than $O(n \log n)$.
 
There are also other families of concatenated codes, based on algebraic-geometry (AG) codes as their outer codes, which can somewhat provide better minimum distance comparing to graph-based codes. More specifically, a certain concatenated code family with a short binary inner code and an AG outer code can surpass the Blokh-Zyablov bound \cite{katsman1984modular}, characterized in \eqref{Blokh-Zyablov}. These codes have construction complexity of $O(n^3\log^3 n)$ \cite{shum2001low} and are currently known to have the largest asymptotic rate, given a certain fixed normalized minimum, while having polynomial construction complexity. Also, the decoding complexity of such codes is $O(n^3)$.

A main drawback of aforementioned construction methods based on concatenation is that they often require $n$ to be very large in order to meet the promised performance. In other words, they exhibit excellent asymptotic performance, however, they often fall short for finite values of $n$ that are of practical interest. Hence, it is desirable to focus on explicit constructions of subspace codes that can be constructed for a wide range of $n$, regardless of how small or large $n$ is, while providing reasonable performance. In a sense, we aim at constructing subspace codes that resemble well-known families of block codes such as Reed-Solomon codes in the subspace domain, and that can be constructed for a broad range of parameters. This is the focus of the next subsection. 


\subsection{A new family of analog subspace codes: Character-polynomial codes}
\label{sec:cp}

In this section, we propose a new family of subspace codes in $G_{1,n}(\C)$ by leveraging polynomial evaluations over finite fields and mapping the finite field symbols to the roots of unity. Then results on character sums from analytic number theory \cite{lidl1997finite}, discussed next, are used to bound the minimum distance of the constructed codes. The resulting codes are referred to as character-polynomial (CP) codes. 

Consider a cyclic group $G$ of order $|G|$. A character $\chi$ associated to $G$ is a homomorphism from $G$ to the unit circle in the complex plane with the regular multiplication of complex numbers, i.e.,
\be
\label{character preserves the product}
\chi(g_1g_2) = \chi(g_1)\chi(g_2),
\ee
for all $g_1,g_2 \in G$. It can be observed that
\be
\label{inverse conjugate}
\chi^*_{g}=\chi(g^{-1}),
\ee
where $g^{-1}$ is the inverse of $g$ in $G$ and $x^*$ is the complex conjugate of $x$ for $x \in \C$.  Given a certain and finite number of characters $\chi_1,\hdots,\chi_l$ one can define the product character $\chi_1\chi_2\hdots \chi_l$ by setting 
$$
(\chi_1\chi_2\hdots \chi_l)(g)=\chi_1(g)\chi_2(g)\hdots \chi_l(g),
$$ 
for all $g \in G$. The set of all  characters associated to $G$ together with this product form an Abelian group of order $|G|$, where the elements $\chi_j$, for $j=1,2,\hdots,|G|$, are described as follows \cite{lidl1997finite}:
\be
\label{character deff for groups}
\chi_j(g'^k)=\text{e}(\frac{jk}{|G|}),
\ee
for $k=0,1,\hdots,|G|-1$, where $g'$ is a generator of $G$ and $\text{e}(x)\,\deff\, \exp(2\pi i x)$. Note that $\chi_0(g)=1$ for $g\in G$ and hence, it is referred to as the \emph{trivial} character. 

A finite field $\Fq$ is naturally equipped with two finite Abelian groups, i.e., the additive and the multiplicative group. Then, the \emph{additive} and \emph{multiplicative} characters of $\Fq$ are defined as the characters associated with the additive and the multiplicative group in $\Fq$, respectively. Using \eqref{character deff for groups}, the additive characters, denoted by $\chi_j$, for $j=0,1,\hdots,q-1$, are described as follows\cite{lidl1997finite}:
\be
\label{character sum deff}
 \chi_j(\alpha )=\text{e}(\frac{\text{tr}_{\text{a}}(j\alpha )}{p})
\ee
for $j \in \Fq$, where $p$ is the characteristic of $\Fq$, and 
$$
\text{tr}_{\text{a}}(\gamma)\,\deff\, \gamma+\gamma^p+\cdots+\gamma^{p^{m-1}}
$$ 
is the \emph{absolute} trace function from $\Fq$ to $\Fp$, where $q=p^m$. Note that \eqref{character sum deff} implies that $\chi_j(\alpha )=\chi_1(j\alpha )$ and the trivial additive character is $\chi_0(\alpha )=1$ for $\alpha \in \Fq$. 

The following result, due to Weil \cite{weil1948some}, on the summations over characters, which are commonly  known  as exponential sums or character sums in the literature, will be utilized in bounding the minimum distance of CP codes, to be discussed next. 
\begin{theorem}\cite[Theorem 5.35]{lidl1997finite}
\label{Weils theorem}
Consider a polynomial $f \in \Fq[x]$ of degree $d\geq 1$ with $\gcd(d,q)=1$. Let $\chi$ be a nontrivial additive character of $\Fq$. Then
\be
\label{weil}
|\sum_{\alpha \in \Fq} \chi(f(\alpha ))|\leq (d-1)\sqrt{q}.
\ee
\end{theorem}

Next, for some $k<q$, let 
\be
\label{Fdef}
\cF \,\deff\, \{f \in \Fq[x]: f(x) = \sum_{i\in [k], i\hspace{-2mm}\mod p \neq 0} f_i x^i \}.
\ee
Note that $|\cF|=q^{ \left \lceil{k(p-1)/p}\right \rceil }$.  We fix $n=q-1$ in our construction.
\begin{definition}
\label{CP_def}
The code $\cC(\cF) \subseteq G_{1,n}(\C)$, referred to as a character-polynomial (CP) code, is defined as follows:
\begin{align}
 \cC(\cF)\,\deff\,\{ \Span{(c_1,c_2,\hdots,c_{n})}: c_i=&\chi\bigl({f(\alpha_i)\bigr)}, \nonumber\\
&\forall f\in \cF, \alpha_i \in \Fq \backslash \{0\} \},\label{cpcode_def}
\vspace{-1mm}
\end{align}
where $\chi$ is a fixed nontrivial additive character of $\Fq$, and $\alpha_i$'s are distinct non-zero elements of $\Fq$. 
\end{definition}

The following theorem provides a lower bound on the normalized minimum distance of $\cC(\cF)$ in terms of $q$ and $d$.
\begin{theorem}
\label{cpcode}
The code $\cC(\cF)$ has size $q^{ \left \lceil{k(p-1)/p}\right \rceil }$ and 
\be
\label{rate_tradeoff}
\delta \geq 1-\frac{\bigl((k-1)\sqrt{q}+1\bigr)^2}{n^2} ,
\ee
where $\delta = d_{\min}/2m$ (here $m=1$) is the normalized minimum distance of the code.
\end{theorem}
\begin{proof} 
Consider distinct codewords $\Span{\bC_1},\Span{\bC_2} \in \cC(\cF)$ with corresponding distinct polynomials $f_1 ,f_2 \in \cF_1$. Let $f = f_2 - f_1$. Then we have 
\begin{align*}
n\norm{\bC_1\bC_2^{\text{H}}}&=|\sum _{\alpha \in\Fq \setminus \{0\}}\chi^*\bigl(f_1(\alpha )\bigr)\chi\bigl(f_2(\alpha )\bigr)|\\
&\stackrel{(a)}{=}|\sum _{\alpha \in\Fq}\chi\bigl(f(\alpha )\bigr)-1|\\
&\stackrel{(b)}{\leq}|\sum _{\alpha \in\Fq}\chi\bigl(f(\alpha )\bigr)|+1\\
&\stackrel{(c)}{\leq}(k-1)\sqrt{q}+1,
\end{align*}
where $(a)$ follows by \eqref{inverse conjugate} and \eqref{character preserves the product} and noting that $\chi\bigl(f(0)\bigr)=1$, $(b)$ is by the triangle inequality, and $(c)$  is by \Tref{Weils theorem} applied to $f = f_1-f_2$. Note that $f \in \cF$. This implies that $\deg(f) \geq 1$ and  $\gcd(\deg(f),q)=1$ since polynomials in $\cF$, as defined in \eq{Fdef}, do not have a monomial of degree divisible by $p$. Hence, $f = f_1-f_2$ satisfies the conditions in \Tref{Weils theorem}.


Using $\Lref{another expression for chordal distance}$, we have
\begin{align*}
\delta&=\frac{d(\Span{\bC_1},\Span{\bC_2})}{2}=1-\norm{\bC_1\bC_2^{\text{H}}}^2>1-\frac{((k-1)\sqrt{q}+1)^2}{n^2},\\
\end{align*}
which completes the proof.
\end{proof}

\begin{figure}[t]
\vspace{-.2mm}
\begin{center}
\includegraphics[width=.9\linewidth]{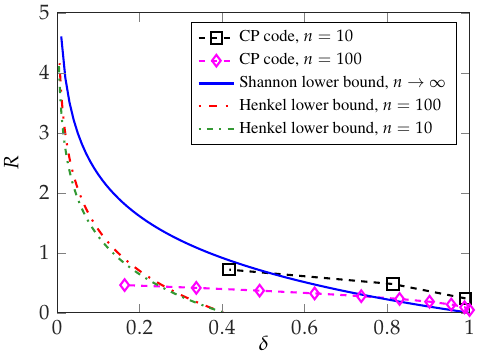}
\caption{\small Comparison of CP codes with lower-bounds in terms of the trade-off between the rate $R$ and the normalized minimum distance $\delta$ for $m=1$.}
\end{center}
\end{figure}

Note that the right hand side of \eq{rate_tradeoff} can be approximated in terms of the rate of the code. This results in a bound on the trade-off between the normalized minimum distance and the rate of the code. We plot this bound and compare it with other bounds/constructions next. In particular, in order to have a simplified numerical analysis, we limit our attention to the case where $q$ is a prime, i.e., $q=p$, and $k < p$. In this case, we have $|\cC(\cF)| = q^k$ and $R = \frac{k \ln q}{n}$. Also, the bound in \eq{rate_tradeoff} is simplified as follows:
\be
\label{rate_tradeoff2}
\delta \geq 1-\frac{((k-1)\sqrt{q}+1)^2}{n^2} > 1 - \frac{qR^2}{(\ln q)^2}.
\ee

It is worth mentioning that all nontrivial choices for the character $\chi$ result in the same codebook. To observe this, let $\chi_a$ and $\chi_b$ denote two distinct nontrivial characters for some $a,b \in \Fq$, and $\cC_a(\cF)$ and $\cC_b(\cF)$ denote the corresponding codebooks  described in Definition\,\ref{CP_def}, respectively. Since both $a$ and $b$ are non-zero, $c= ba^{-1}$ is a well-defined non-zero element of $\Fq$. Then, for any $f(x) \in \cF$, one can write
\begin{align}
\chi_b(f(x))=&\text{e}(\frac{\text{tr}_{\text{a}}(bf(x) )}{p})=\text{e}(\frac{\text{tr}_{\text{a}}(ac f(x) )}{p}) \nonumber \\
=& \text{e}(\frac{\text{tr}_{\text{a}}(a f'(x) )}{p})=\chi_a(f'(x)). \label{equvalence}
\end{align}

Note that $f'(x)\deff cf(x)$ is also in $\cF$. This together with \eqref{equvalence} imply that $\cC_a(\cF)=\cC_b(\cF)$. 

In Figure\,2, we compare the trade-off between the rate $R$ and the normalized minimum distance $\delta$ that the CP codes offer at different values of $n$ with Shannon's lower bound \cite{shannon1959probability}, for $n \rightarrow \infty$, and lower bounds derived by Henkel \cite[Theorem 4.2]{henkel2005sphere} for finite values of $n$. Note that these lower bounds are of the same type as Gilbert-Varshamov bound and do not yield explicit constructions. Nevertheless, it can be observed that CP codes outperform these lower bounds at low rates, thereby improving these bounds while providing explicit constructions. Note also that the trade-off between $R$ and $\delta$ shown in Figure\,2 for CP codes is derived from the bound established in \Tref{cpcode}. In other words, the actual value of $\delta$ for the given values of $R$ can be larger than what is shown in Figure\,2.

\noindent
{\bf Remark\,5.} Given a subspace code in $\cC \subseteq G_{m,n}(\C)$ one can construct a code in $G_{2m,2n}(\R)$ by mapping  $\bC_i \in \cC$ to   
\be\label{map}
\left[
\begin{array}{c c}
\Re(\bC_i) & \Im(\bC_i) \\
-\Im(\bC_i) & \Re(\bC_i)
\end{array}
\right],
\ee
where $\Re$ and $\Im$ represent the real part and the imaginary part of their input, respectively. It can be observed that this mapping preserves the normalized distance between the codewords. Hence, the normalized minimum distance of $\cC$ is also preserved. This mapping enables us to construct codes in $G_{2,n}(\R)$ using the proposed CP codes, while keeping the normalized minimum distance and the size of the code the same, in order to have fair comparisons with existing code constructions in the real Grassmann space. 

\begin{figure}[t]
\vspace{1.5mm}
\begin{center}
\includegraphics[width=.9 \linewidth]{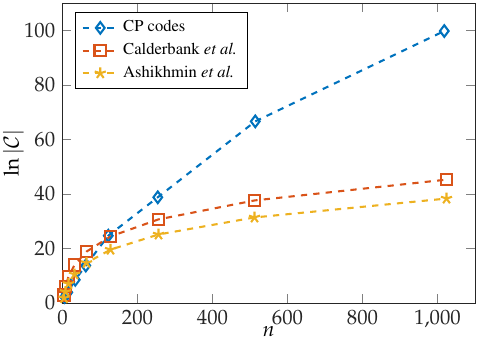}
\caption{\small Comparison of the codes  in $G_{2,n}(\R)$ obtained from our proposed CP  codes in $G_{1,\frac{n}{2}}(\C)$ with the codes constructed by Calderbank \textit{et al.} \cite{calderbank1999group} and Ashikhmin \textit{et al.} \cite{ashikhmin2010grassmannian}. The codes from \cite{calderbank1999group} \cite{ashikhmin2010grassmannian} have $\delta = \frac{1}{2}$, and for CP codes we have $\delta \geq \frac{1}{2}$, for all considered values of $n$.}\label{one-dim-comparison}
\end{center}
\end{figure}

In Figure\,3, we compare CP codes with two existing constructions of Grassmann codes, that are constructed explicitly for a wide range of $n$, in the literature. In \cite{calderbank1999group},  Calderbank \emph{et al.} introduce a group-theoretic framework for packing in $G_{2^i,2^k}(\R)$ for any pair of integers $(i,k)$ with $i\leq k$. In another prior work, Ashikhmin \emph{et al.} \cite{ashikhmin2010grassmannian} provide a code construction method in $G_{2^i,2^k}(\C)$ based on binary Reed-Muller (RM) codes.  It is worth noting that the subspaces in this construction correspond to certain projection operators determined by Pauli matrices appearing in quantum computing, where the main objective is to enable error correction in quantum computing. Therefore, one should not expect such constructions to maximize the rate. However, they are competitive in small dimensions as illustrated in Figure\,\ref{one-dim-comparison}.    
By utilizing the mapping specified in \eqref{map} for both the CP codes and the codes constructed in \cite{ashikhmin2010grassmannian} with $i=0$, we compare the log-size of the codes obtained in $G_{2,n}(\R)$ with that of codes in $G_{2,n}(\R)$ from  \cite{calderbank1999group}, while fixing $\delta = \frac{1}{2}$ for the codes from \cite{calderbank1999group} \cite{ashikhmin2010grassmannian} and having $\delta \geq \frac{1}{2}$ for the CP codes, for all the considered values of $n$. The results are shown in Figure \,3. Note that $n$ is equal to $2^k$ for the constructions in \cite{calderbank1999group} and \cite{ashikhmin2010grassmannian}, while for CP codes we pick $n = 2p$, where $p$ is the largest $p$ with $p < 2^{k-1}$, for $k\in \{3,\hdots,10\}$. It can be observed that CP codes offer significantly larger code size and, consequently, rate comparing to the other explicit constructions, as $n$ grows large.         

\noindent
{\bf Remark\,6.} Note that the lower bound of \Tref{cpcode} on $\dmin$ of CP codes is at most two. Even with $d_{\min}$ slightly greater than two \Tref{error correction capability of codes related to minimum distance} implies that the correction of only one error can be guaranteed using a minimum distance decoder. A possible solution to this issue is to consider list decoders in order to guarantee error recovery beyond $d_{\min}/2$. In the finite field domain, several prior works have studied list decoding for algebraic subspace codes, see, e.g., \cite{MV1,GNW, MV2, GWX, MV3}. In particular, it is observed in \cite{MV2} that an unbounded number of errors can be corrected by increasing the list size, while the dimension of subspace codewords is one. However, obtaining similar results in the analog domain necessitates further investigation and is left for future work.

\noindent
{\bf Remark\,7.} The Weil bound, recapped in \Tref{Weils theorem}, has been utilized in various coding theoretic contexts , e.g., to provide bounds on the minimum distance of the duals of BCH codes \cite{10.2307/2099415} and to estimate the covering radius of BCH codes with large block lengths \cite{helleseth1985covering,tietainen1987covering}. Also, it has inspired the design of certain families of sequences over finite fields of prime size with low auto-correlation and cross-correlation in \cite{blake1982note} in a similar fashion. Notably, Kumar \emph{et al.} derive a counterpart of the Weil bound over Galios rings \cite{kumar1995upper}. They further design families of phase-shift-keying sequences with low correlation where this bound is then utilized to provide a guarantee on the correlation level of the construction. The main difference between the design criterion considered in the constructions in \cite{blake1982note} and \cite{kumar1995upper} and that of our approach is that these prior works require the correlation between the circular shifts of any two codewords to be bounded, in addition to the auto-correlation, while we only require the correlation between two distinct codewords to be small. Furthermore, as it is shown in the next section, a certain property of the construction provided in this work is leveraged to construct subspace codes in higher dimensions ($m>1$).

\subsection{Higher Dimensional Character-Polynomial Codes}

The \emph{character-polynomial} (CP) codes, demonstrated in Section\,\ref{sec:cp}, provide a family of one-dimensional subspaces in a complex Grassmann space, i.e., a packing of lines in $G_{1,n}(\C)$. Next, we provide a slightly different version of one-dimensional CP codes that enables us to provide a new construction in $G_{m,n}(\C)$ for $m>1$.  We fix $n=q$ in this section.
 \begin{definition}\label{CP_def2}
The code $\cC'(\cF) \subseteq G_{1,n}(\C)$ is defined as follows:
\be\label{cpcode_def2}
\cC'(\cF)\,\deff\,\{\Span{ (c_1,c_2,\hdots,c_{n})}: c_i=\chi\bigl({f(\alpha_i)\bigr)},\forall f\in \cF, \alpha_i \in \Fq\},
\vspace{-1mm}
\ee
where $\chi$ is a fixed nontrivial additive character of $\Fq$, $\cF$ is defined in \eq{Fdef}, and $\alpha_i$'s are distinct elements of $\Fq$.
\end{definition}

The  definition of CP codes provided in \eqref{cpcode_def} excludes the zero element of $\Fq$ from the set of evaluation points. Including the zero element in the alternative version, specified in \eqref{cpcode_def2}, leads to a certain property that is discussed in the following lemma. Before that, we define the following. For any two sets of orthonormal bases $B_1$ and $B_2$ for $W$, the \emph{cross-correlation} between $B_1$ and $B_2$ is defined as 
\be{}
\Delta_{B_1,B_2}\deff \max_{\bv_1\in B_1, \bv_2 \in B_2} |v_1\cdot v_2|,
\ee
where the operation $\cdot$ denotes the inner product.

\begin{lemma}\label{mutual-cor}
The set of normal vectors representing revised one-dimensional CP codewords, defined in Definition\,\ref{CP_def2}, can be split into $q^{ \left \lceil{k(p-1)/p}\right \rceil -1}$, denoted by $b$, collections $B_i$'s, for $i\in[b]$, where each $B_i$ is an orthonormal basis for $W$. Furthermore, the cross-correlation between $B_i$'s is upper bounded as follows: 
\be\label{mutual-cor-bound}
\max_{i,j \in [b], i\neq j} \Delta_{B_i,B_j}\leq\frac{(k-1)}{\sqrt{n}}.
\ee
\end{lemma}
\begin{proof}
The set of polynomials $\cF$, defined in \eqref{Fdef}, can be split into disjoint subsets such that the polynomials belonging to the same subset differ only in the coefficient of the degree-one monomial, i.e., the coefficient of $x$. Note that the constant coefficient of all the polynomials in $\cF$ is equal to zero  according to \eqref{Fdef}. Then, two distinct polynomials $f$ and $f'$ in $\cF$ belong to the same subset if and only if $\deg(f-f')=1$. Consequently, it can be observed that $\cF$ is partitioned into $q^{ \left \lceil{k(p-1)/p}\right \rceil -1}$ of such subsets each of size $q$. Let $\bc=\frac{1}{\sqrt{n}}(c_1, \cdots, c_n)$ and $\bc'=\frac{1}{\sqrt{n}}(c_1', \cdots, c_n')$, where $c_i=\chi\bigl(f(\alpha_i)\bigr)$ and $c'_i=\chi\bigl(f'(\alpha_i)\bigr)$ for $i \in [n]$, and $f$ and $f'$ belong to the subset of $\cF$ as discussed above. Then,
\be\label{aa}\small
|\bc \cdot \bc'|= \frac{1}{n}|\sum _{\alpha \in\Fq} \chi^*\bigl(f(\alpha )\bigr)\chi\bigl(f'(\alpha )\bigr)| \stackrel{(a)}{=} \frac{1}{n}|\sum _{\alpha \in\Fq}\chi\bigl((f'-f)(\alpha )\bigr)| \stackrel{(b)}{\leq} 0
\ee
where $(a)$ follows by \eqref{character sum deff} and $(b)$ is by the Weil bound, specified in \eqref{weil}, together with noting that $\deg(f-f')=1$. Therefore, \eqref{aa} implies that we must have $\bc \cdot \bc' = 0$, i.e., $\bc$ and $\bc'$ must be orthogonal. Hence, each of the subsets of $\cF$, as discussed above, consists of $q$ mutually orthogonal lines in $W$. Consequently, the set of unit-norm vectors representing these lines is an orthonormal basis for $W$. The upper bound in \eqref{mutual-cor-bound} can be derived again using the Weil bound and by noting that $\deg(f-f')\leq k$ for any $f$ and $f'$ in $\cF$, i.e., for any two normalized distinct $\bc$ and $\bc'$ in the CP code, the upper bound in \eqref{mutual-cor-bound} holds. 
\end{proof}

Inspired by \Lref{mutual-cor}, we provide a construction for packing $m$-planes in $G_{m,n}(\C)$ for $m>1$. Let $\bv^{(i)}_{1}, \cdots, \bv^{(i)}_{q}$ denote the orthonormal basis vectors in $B_i$, for $i \in [b]$, where $b$ is defined in \Lref{mutual-cor}. Also, let
\be\label{phi-ij}
\Phi_{ij}=\begin{bmatrix}
\bv^{(i)}_{(j-1)m+1}\\
\vdots\\
\bv^{(i)}_{jm}
\end{bmatrix},
\ee
for all $i \in [b]$ and $j \in [\floor{\frac{q}{m}}]$.
Then,
\be\label{high-dim-code}
\cC\,\deff\,\{\Span{\Phi_{ij}}:\forall i \in [b], \forall j \in [\floor{\frac{q}{m}}] \}
\ee
is a subspace code in $G_{m,n}(\C)$. Note that we have
$$
|\cC|=q^{ \left \lceil{k(p-1)/p}\right \rceil -1}\floor{\frac{q}{m}}.
$$

The normalized minimum distance of $\cC$ is characterized in the next theorem. 
\begin{theorem}
The normalized minimum distance $\delta$ of the code $\cC$, defined in \eq{high-dim-code}, is lower bounded as
\be\label{dmin}
\delta\geq {1-\frac{m(k-1)^2}{n}}.
\ee
\end{theorem}
\begin{proof}
Consider two distinct codewords $C_1=\Span{\Phi_{i,j}}$ and $C_2=\Span{\Phi_{i',j'}} \in \cC$. Note that the rows in $\Phi_{i,j}$ and $\Phi_{i',j'}$ are orthonormal. Note also that for $i=i'$, $\Phi_{i,j}\Phi_{i',j'}^{\text{H}}=\boldsymbol{0}$, since all the rows of both matrices belong to the same orthonormal basis which implies that $\delta=1$ in this case. Otherwise, i.e., when $i\neq i'$, we have 

\begin{equation}\label{phi_bound}
    \norm{\Phi_{i,j}\Phi_{i',j'}^{\text{H}}}^2\leq \frac{m^2 (k-1)^2}{n},
\end{equation}
since $\Phi_{i,j}\Phi_{i',j'}^{\text{H}}$ is an $m \times m$ matrix whose elements' absolute value is upperbounded by  \eqref{mutual-cor-bound}. Then, one can write 


\begin{align}
    \delta &= \frac{d(C_1,C_2)}{2m} \label{normalaized_dist_deff_1}\\
    &=\frac{2(m-\norm{\Phi_{i,j}\Phi_{i',j'}^{\text{H}}}^2)}{2m}\label{alter_characterization_2}\\
    & \geq 1-\frac{m(k-1)^2}{n}\label{phi_char_3},
\end{align}
where \eqref{normalaized_dist_deff_1} is by the definition of the normalized distance provided in Definition\,\ref{code parameters}, \eqref{alter_characterization_2} follows by utilizing the  alternative characterization provided in \Lref{another expression for chordal distance} for  the distance function defined in \eqref{d function}, and \eqref{phi_char_3} results from  plugging in \eqref{phi_bound} into \eqref{alter_characterization_2}.
\end{proof}

In Table\,\ref{calderbank-table}, we compare the parameters of our proposed codes with those of the closest explicit construction of Grassmann codes proposed in \cite{calderbank1999group}. 
By utilizing the mapping specified in \eqref{map}  for
our codes in $G_{\frac{m}{2},\frac{n}{2}}(\C)$,
we compare the blocklength, logarithm of the code size, and the minimum distance  of the codes obtained in $G_{m,n}(\R)$ with those of the codes in $G_{m,n}(\R)$ from  \cite{calderbank1999group}. In all the instances of $n$ and $m=4,8,16,32$ in Table\,\ref{calderbank-table}, the normalized minimum distance is equal to $\frac{1}{{2}}$ for codes from \cite{calderbank1999group} 
while it is at least $\frac{1}{{2}}$ for our codes. Note that $n$ is equal to $2^k$ for the construction in \cite{calderbank1999group}
, while for our codes we pick $n = 2p$, where $p$ is the largest prime number with $p < 2^{k-1}$, for various choices of $k$. 
It can be observed that our proposed  construction offer significantly larger code size and, consequently, rate comparing to the explicit construction of \cite{calderbank1999group}, as $n$ grows large. Note also that even for small $n$ in a few rows of Table\,\ref{calderbank-table} where the log-size of the proposed CP codes is not larger than that of the codes constructed in \cite{calderbank1999group} yet, the minimum distance offered by CP codes is still larger while having a smaller blocklength. The other advantage of our construction is that it does not have any constraint on the codeword dimension $m$ but $m$ has to be a power of two in \cite{calderbank1999group}.

\begin{table}[]
	\centering
	\begin{tabular}{ |c|ccc|ccc|} 
	        \hline
	        & \multicolumn{3}{|c|}{Our construction } 
	         &\multicolumn{3}{|c|}{Calderbank \emph{et al.} \cite{calderbank1999group}}\\
		 \hline
		 $m$ & $n$ & $\ln(|\cC|)$& $d_{\min}$ & $n$ & $\ln(|\cC|)$ & $d_{\min}$ \\
		\hline
		 $4$ & $254$ & $28.36$ & $2.43$& $256$ &$32.62$ & $2$\\ 
		
		$4$& $502$ & $43.51$& $2.44$ & $512$ &$40.25$ & $2$\\ 
		
		$4$ & $1018$ & $74.09$ & $2.10$& $1024$ & $48.57$ & $2$\\ 
		
		$4$& $2042$ & $110.2$ &$2.37$ & $2048$ & $57.59 $ & $2$\\

		
		$8$ & $1018$ & $48.47$ & $4.92$& $1024$ & $49.87$ & $4$\\ 
		
		$8$ & $2024$ & $81.76$ & $4.3$ & $2048$ & $59.57$ &$4$\\
		
		$8$ & $4078$ & $120.5$ & $4.47$& $4096$ & $69.97$ &$4$\\ 
		
		$8$ & $8186$ & $189.9$& $4.22$ & $8192$ & $81.07$ & $4$\\

		
		$16$ & $2042$ & $53.34$ & $9.86$& $2048$ & $59.52$ & $8$\\ 
		
		$16$ & $4078$ & $89.36$ & $8.40$ & $4096$ & $70.61$ & $8$\\
		
		
		$32$ & $4078$ & $58.19$ & $19.67$& $4096$ & $69.19$ & $16$ \\ 
		
		$32$ & $8186$ & $97.03$ &$16.86$ & $8192$ & $81.67$ & $16$\\
	    \hline	
	
	\end{tabular}
	\caption{\small  Parameters of the new construction in Grassmannian space provided in this paper to those of the proposed scheme by Calderbank \emph{et al.} \cite[Theorem\,1]{calderbank1999group}. }\label{calderbank-table}
	\vspace{-5mm}
\end{table}

\section{Conclusion and Future Work}
\label{sec:six}

In this paper, motivated by the emergence of massive wireless networks, we provided a new coding framework for non-coherent communications across such networks in order to mitigate the network deficiency and interference, e.g., from neighbouring cells in a cellular network. To this end, the concept of analog operator channel was introduced that captures the effect of network deficiency and interference as subspace erasure and errors, respectively. Also, a new distance is defined and relations between the error-and-erasure correction capability of a subspace code in the analog domain and its minimum distance is established. This leads to a code design criteria to correct errors/erasures over the analog operator channel. Furthermore, we extended the framework to the case with additive noise, that naturally exists in physical layer links, and showed that the analog operator channel is robust with respect to the additive noise. As a consequence, the effect of noise is shown to be negligible in high signal-to-noise ratio regimes. Finally, we proposed a novel algebraic construction for  subspace codes in the complex domain that outperforms the existing constructions in the literature, in terms of the rate-minimum distance trade-off, for a wide range of code blocklength. 

There are several directions for future research. 
Extending the results derived for minimum distance decoder by exploring list decoding algorithms is a natural direction for future research.
In the case of noisy operator channel, obtaining tighter bounds on the subspace perturbation imposed by the additive noise is another research direction.         

\bibliographystyle{IEEEtran}
\bibliography{ref}

\begin{thebibliography}{10}
\providecommand{\url}[1]{#1}
\csname url@samestyle\endcsname
\providecommand{\newblock}{\relax}
\providecommand{\bibinfo}[2]{#2}
\providecommand{\BIBentrySTDinterwordspacing}{\spaceskip=0pt\relax}
\providecommand{\BIBentryALTinterwordstretchfactor}{4}
\providecommand{\BIBentryALTinterwordspacing}{\spaceskip=\fontdimen2\font plus
\BIBentryALTinterwordstretchfactor\fontdimen3\font minus
  \fontdimen4\font\relax}
\providecommand{\BIBforeignlanguage}[2]{{%
\expandafter\ifx\csname l@#1\endcsname\relax
\typeout{** WARNING: IEEEtran.bst: No hyphenation pattern has been}%
\typeout{** loaded for the language `#1'. Using the pattern for}%
\typeout{** the default language instead.}%
\else
\language=\csname l@#1\endcsname
\fi
#2}}
\providecommand{\BIBdecl}{\relax}
\BIBdecl

\bibitem{3gpp-5g-2}
``{5G}; {Study} on new radio {(NR)} access technology physical layer aspects,''
  3GPP TR 38.802, March 2017.

\bibitem{nakamura2013trends}
T.~Nakamura, S.~Nagata, A.~Benjebbour, Y.~Kishiyama, T.~Hai, S.~Xiaodong,
  Y.~Ning, and L.~Nan, ``Trends in small cell enhancements in {LTE} advanced,''
  \emph{IEEE Communications Magazine}, vol.~51, no.~2, pp. 98--105, 2013.

\bibitem{sun2013interference}
S.~Sun, Q.~Gao, Y.~Peng, Y.~Wang, and L.~Song, ``Interference management
  through {CoMP} in {3GPP} {LTE}-advanced networks,'' \emph{IEEE Wireless
  Communications}, vol.~20, no.~1, pp. 59--66, 2013.

\bibitem{deb2014algorithms}
S.~Deb, P.~Monogioudis, J.~Miernik, and J.~P. Seymour, ``Algorithms for
  enhanced inter-cell interference coordination {(eICIC)} in {LTE} {HetNets},''
  \emph{IEEE/ACM Transactions on Networking (ToN)}, vol.~22, no.~1, pp.
  137--150, 2014.

\bibitem{KK}
R.~Koetter and F.~R. Kschischang, ``Coding for errors and erasures in random
  network coding,'' \emph{IEEE Transactions on Information Theory}, vol.~54,
  no.~8, pp. 3579--3591, 2008.

\bibitem{ho2003benefits}
T.~Ho, R.~Koetter, M.~Medard, D.~R. Karger, and M.~Effros, ``The benefits of
  coding over routing in a randomized setting,'' \emph{Proceedings of IEEE
  International Symposium on Information Theory}, 2003.

\bibitem{shannon1959probability}
C.~E. Shannon, ``Probability of error for optimal codes in a {Gaussian}
  channel,'' \emph{Bell System Technical Journal}, vol.~38, no.~3, pp.
  611--656, 1959.

\bibitem{barg2002bounds}
A.~Barg and D.~Y. Nogin, ``Bounds on packings of spheres in the {Grassmann}
  manifold,'' \emph{IEEE Transactions on Information Theory}, vol.~48, no.~9,
  pp. 2450--2454, 2002.

\bibitem{bachoc2006linear}
C.~Bachoc, ``Linear programming bounds for codes in {Grassmannian} spaces,''
  \emph{IEEE Transactions on Information Theory}, vol.~52, no.~5, pp.
  2111--2125, 2006.

\bibitem{bachoc2006bounds}
C.~Bachoc, Y.~Ben-Haim, and S.~Litsyn, ``Bounds for codes in the {Grassmann}
  manifold,'' in \emph{2006 IEEE 24th Convention of Electrical \& Electronics
  Engineers in Israel}.\hskip 1em plus 0.5em minus 0.4em\relax IEEE, 2006, pp.
  25--29.

\bibitem{barg2006bound}
A.~Barg and D.~Nogin, ``A bound on {Grassmannian} codes,'' \emph{Journal of
  Combinatorial Theory, Series A}, vol. 113, no.~8, pp. 1629--1635, 2006.

\bibitem{marzetta1999capacity}
T.~L. Marzetta and B.~M. Hochwald, ``Capacity of a mobile multiple-antenna
  communication link in {Rayleigh} flat fading,'' \emph{IEEE Transactions on
  Information Theory}, vol.~45, no.~1, pp. 139--157, 1999.

\bibitem{hochwald2000unitary}
B.~M. Hochwald and T.~L. Marzetta, ``Unitary space-time modulation for
  multiple-antenna communications in {Rayleigh} flat fading,'' \emph{IEEE
  Transactions on Information Theory}, vol.~46, no.~2, pp. 543--564, 2000.

\bibitem{zheng2002communication}
L.~Zheng and D.~N.~C. Tse, ``Communication on the {Grassmann} manifold: A
  geometric approach to the noncoherent multiple-antenna channel,'' \emph{IEEE
  Transactions on Information Theory}, vol.~48, no.~2, pp. 359--383, 2002.

\bibitem{cox2012introduction}
C.~Cox, \emph{An introduction to {LTE}: {LTE}, {LTE}-advanced, {SAE} and {4G}
  mobile communications}.\hskip 1em plus 0.5em minus 0.4em\relax John Wiley \&
  Sons, 2012.

\bibitem{henkel2005sphere}
O.~Henkel, ``Sphere-packing bounds in the {Grassmann} and {Stiefel}
  manifolds,'' \emph{IEEE Transactions on Information Theory}, vol.~51, no.~10,
  pp. 3445--3456, 2005.

\bibitem{conway1996packing}
J.~H. Conway, R.~H. Hardin, and N.~J. Sloane, ``Packing lines, planes, etc.:
  Packings in {Grassmannian} spaces,'' \emph{Experimental mathematics}, vol.~5,
  no.~2, pp. 139--159, 1996.

\bibitem{dhillon2008constructing}
I.~S. Dhillon, J.~R. Heath, T.~Strohmer, and J.~A. Tropp, ``Constructing
  packings in {Grassmannian} manifolds via alternating projection,''
  \emph{Experimental mathematics}, vol.~17, no.~1, pp. 9--35, 2008.

\bibitem{nebe2001invariants}
G.~Nebe, E.~M. Rains, and N.~J. Sloane, ``The invariants of the {Clifford}
  groups,'' \emph{Designs, Codes and Cryptography}, vol.~24, no.~1, pp.
  99--122, 2001.

\bibitem{shor1998family}
P.~Shor and N.~J.~A. Sloane, ``A family of optimal packings in {Grassmannian}
  manifolds,'' \emph{Journal of Algebraic Combinatorics}, vol.~7, no.~2, pp.
  157--163, 1998.

\bibitem{calderbank1999group}
A.~Calderbank, R.~Hardin, E.~Rains, P.~Shor, and N.~J.~A. Sloane, ``A
  group-theoretic framework for the construction of packings in {Grassmannian}
  spaces,'' \emph{Journal of Algebraic Combinatorics}, vol.~9, no.~2, pp.
  129--140, 1999.

\bibitem{calderbank1997quantum}
A.~R. Calderbank, E.~M. Rains, P.~W. Shor, and N.~J. Sloane, ``Quantum error
  correction and orthogonal geometry,'' \emph{Physical Review Letters},
  vol.~78, no.~3, p. 405, 1997.

\bibitem{calderbank1998quantum}
A.~R. Calderbank, E.~M. Rains, P.~Shor, and N.~J. Sloane, ``Quantum error
  correction via codes over {GF}(4),'' \emph{IEEE Transactions on Information
  Theory}, vol.~44, no.~4, pp. 1369--1387, 1998.

\bibitem{sloane2002packing}
N.~J.~A. Sloane, ``Packing planes in four dimensions and other mysteries,''
  \emph{arXiv preprint math/0208017}, 2002.

\bibitem{itw2019}
O.~Tirkkonen1 and R.~Calderbank, ``Codebooks of complex lines based on binary
  subspace chirps,'' in \emph{Proceedings of Information Theory Workshop
  (ITW)}, 2019.

\bibitem{strohmer2003grassmannian}
T.~Strohmer and R.~W. Heath~Jr, ``Grassmannian frames with applications to
  coding and communication,'' \emph{Applied and computational harmonic
  analysis}, vol.~14, no.~3, pp. 257--275, 2003.

\bibitem{xia2005achieving}
P.~Xia, S.~Zhou, and G.~B. Giannakis, ``Achieving the {Welch} bound with
  difference sets,'' \emph{IEEE Transactions on Information Theory}, vol.~51,
  no.~5, pp. 1900--1907, 2005.

\bibitem{tropp2005designing}
J.~A. Tropp, I.~S. Dhillon, R.~W. Heath, and T.~Strohmer, ``Designing
  structured tight frames via an alternating projection method,'' \emph{IEEE
  Transactions on information theory}, vol.~51, no.~1, pp. 188--209, 2005.

\bibitem{kutyniok2009robust}
G.~Kutyniok, A.~Pezeshki, R.~Calderbank, and T.~Liu, ``Robust dimension
  reduction, fusion frames, and {Grassmannian} packings,'' \emph{Applied and
  Computational Harmonic Analysis}, vol.~26, no.~1, pp. 64--76, 2009.

\bibitem{ding2007generic}
C.~Ding and T.~Feng, ``A generic construction of complex codebooks meeting the
  {Welch} bound,'' \emph{IEEE Transactions on information theory}, vol.~53,
  no.~11, pp. 4245--4250, 2007.

\bibitem{han2006geometrical}
G.~Han and J.~Rosenthal, ``Geometrical and numerical design of structured
  unitary space--time constellations,'' \emph{IEEE transactions on information
  theory}, vol.~52, no.~8, pp. 3722--3735, 2006.

\bibitem{hochwald2000systematic}
B.~M. Hochwald, T.~L. Marzetta, T.~J. Richardson, W.~Sweldens, and R.~Urbanke,
  ``Systematic design of unitary space-time constellations,'' \emph{IEEE
  Transactions on Information Theory}, vol.~46, no.~6, pp. 1962--1973, 2000.

\bibitem{agrawal2001multiple}
D.~Agrawal, T.~J. Richardson, and R.~L. Urbanke, ``Multiple-antenna signal
  constellations for fading channels,'' \emph{IEEE Transactions on Information
  Theory}, vol.~47, no.~6, pp. 2618--2626, 2001.

\bibitem{aggarwal2006grassmannian}
V.~Aggarwal, A.~Ashikhmin, and A.~R. Calderbank, ``A {Grassmannian} packing
  based on the {Nordstrom-Robinson} code,'' in \emph{Proceedings of IEEE
  Information Theory Workshop}.\hskip 1em plus 0.5em minus 0.4em\relax IEEE,
  2006, pp. 1--5.

\bibitem{ashikhmin2010grassmannian}
A.~Ashikhmin and A.~R. Calderbank, ``Grassmannian packings from operator
  {Reed--Muller} codes,'' \emph{IEEE Transactions on Information Theory},
  vol.~56, no.~11, pp. 5689--5714, 2010.

\bibitem{reboredo2014compressive}
H.~Reboredo, F.~Renna, R.~Calderbank, and M.~R. Rodrigues, ``Compressive
  classification of a mixture of gaussians: Analysis, designs and geometrical
  interpretation,'' \emph{arXiv preprint arXiv:1401.6962}, 2014.

\bibitem{nokleby2015discrimination}
M.~Nokleby, M.~Rodrigues, and R.~Calderbank, ``Discrimination on the
  {Grassmann} manifold: Fundamental limits of subspace classifiers,''
  \emph{IEEE Transactions on Information Theory}, vol.~61, no.~4, pp.
  2133--2147, 2015.

\bibitem{nokleby2013information}
M.~Nokleby, R.~Calderbank, and M.~R. Rodrigues, ``Information-theoretic limits
  on the classification of {Gaussian} mixtures: Classification on the grassmann
  manifold,'' in \emph{2013 IEEE Information Theory Workshop (ITW)}.\hskip 1em
  plus 0.5em minus 0.4em\relax IEEE, 2013, pp. 1--5.

\bibitem{huang2015role}
J.~Huang, Q.~Qiu, and R.~Calderbank, ``The role of principal angles in subspace
  classification,'' \emph{IEEE Transactions on Signal Processing}, vol.~64,
  no.~8, pp. 1933--1945, 2015.

\bibitem{roy1947note}
S.~Roy, ``A note on critical angles between two flats in hyperspace with
  certain statistical applications,'' \emph{Sankhy{\=a}: The Indian Journal of
  Statistics}, pp. 177--194, 1947.

\bibitem{ye2014distance}
K.~Ye and L.-H. Lim, ``Distance between subspaces of different dimensions,''
  \emph{arXiv preprint arXiv:1407.0900}, vol.~4, 2014.

\bibitem{sun1991perturbation}
J.-G. Sun, ``Perturbation bounds for the {Cholesky} and {QR} factorizations,''
  \emph{BIT Numerical Mathematics}, vol.~31, no.~2, pp. 341--352, 1991.

\bibitem{golub1996matrix}
H.~Golub and C.~F. Van~Loan, ``Matrix computations,'' \emph{Press, London},
  1996.

\bibitem{hwang2004cauchy}
S.-G. Hwang, ``Cauchy's interlace theorem for eigenvalues of {Hermitian}
  matrices,'' \emph{The American Mathematical Monthly}, vol. 111, no.~2, pp.
  157--159, 2004.

\bibitem{elias}
P.~Elias, ``Error-free coding,'' \emph{IRE Trans. on Inform. Theory,}, pp.
  29--37, 1954.

\bibitem{forney1965concatenated}
G.~D. Forney, ``Concatenated codes.'' 1965.

\bibitem{zyablov1971estimate}
V.~V. Zyablov, ``An estimate of the complexity of constructing binary linear
  cascade codes,'' \emph{Problemy Peredachi Informatsii}, vol.~7, no.~1, pp.
  5--13, 1971.

\bibitem{blokh1982linear}
E.~Blokh and V.~V. Zyablov, ``Linear concatenated codes,'' \emph{Moscow, USSR:
  Nauka}, 1982.

\bibitem{tanner1981recursive}
R.~Tanner, ``A recursive approach to low complexity codes,'' \emph{IEEE
  Transactions on information theory}, vol.~27, no.~5, pp. 533--547, 1981.

\bibitem{barg2006distance}
A.~Barg and G.~Zemor, ``Distance properties of expander codes,'' \emph{IEEE
  Transactions on Information Theory}, vol.~52, no.~1, pp. 78--90, 2006.

\bibitem{katsman1984modular}
G.~Katsman, M.~Tsfasman, and S.~Vladut, ``Modular curves and codes with a
  polynomial construction,'' \emph{IEEE Transactions on Information Theory},
  vol.~30, no.~2, pp. 353--355, 1984.

\bibitem{shum2001low}
K.~W. Shum, I.~Aleshnikov, P.~V. Kumar, H.~Stichtenoth, and V.~Deolalikar, ``A
  low-complexity algorithm for the construction of algebraic-geometric codes
  better than the {Gilbert-Varshamov} bound,'' \emph{IEEE Transactions on
  Information Theory}, vol.~47, no.~6, pp. 2225--2241, 2001.

\bibitem{lidl1997finite}
R.~Lidl and H.~Niederreiter, \emph{Finite fields}.\hskip 1em plus 0.5em minus
  0.4em\relax Cambridge university press, 1997, vol.~20.

\bibitem{weil1948some}
A.~Weil, ``On some exponential sums,'' \emph{Proceedings of the National
  Academy of Sciences of the United States of America}, vol.~34, no.~5, p. 204,
  1948.

\bibitem{MV1}
H.~Mahdavifar and A.~Vardy, ``Algebraic list-decoding on the operator
  channel,'' \emph{Proceedings of IEEE International Symposium on Information
  Theory}, pp. 1193--1197, 2010.

\bibitem{GNW}
V.~Guruswami, S.~Narayanan, and C.~Wang, ``List decoding subspace codes from
  insertions and deletions,'' \emph{Proceedings of the 3rd Innovations in
  Theoretical Computer Science Conference}, pp. 183--189, 2012.

\bibitem{MV2}
H.~Mahdavifar and A.~Vardy, ``Algebraic list-decoding of subspace codes,''
  \emph{IEEE Transactions on Information Theory}, vol.~59, no.~12, pp.
  7814--7828, 2013.

\bibitem{GWX}
V.~Guruswami, C.~Wang, and C.~Xing, ``Explicit list-decodable rank-metric and
  subspace codes via subspace designs,'' \emph{IEEE Transactions on Information
  Theory}, vol.~62, no.~5, pp. 2707--2718, May 2016.

\bibitem{MV3}
H.~Mahdavifar and A.~Vardy, ``Algebraic list-decoding in projective space:
  Decoding with multiplicities and rank-metric codes,'' \emph{IEEE Transactions
  on Information Theory}, vol.~65, no.~2, pp. 1085--1100, 2019.

\bibitem{10.2307/2099415}
\BIBentryALTinterwordspacing
D.~R. Anderson, ``A new class of cyclic codes,'' \emph{SIAM Journal on Applied
  Mathematics}, vol.~16, no.~1, pp. 181--197, 1968. [Online]. Available:
  \url{http://www.jstor.org/stable/2099415}
\BIBentrySTDinterwordspacing

\bibitem{helleseth1985covering}
T.~Helleseth, ``On the covering radius of cyclic linear codes and arithmetic
  codes,'' \emph{Discrete Applied Mathematics}, vol.~11, no.~2, pp. 157--173,
  1985.

\bibitem{tietainen1987covering}
A.~Tiet{\"a}inen, ``On the covering radius of long binary {BCH} codes,''
  \emph{Discrete Applied Mathematics}, vol.~16, no.~1, pp. 75--77, 1987.

\bibitem{blake1982note}
I.~Blake and J.~Mark, ``A note on complex sequences with low correlations
  (corresp.),'' \emph{IEEE Transactions on Information Theory}, vol.~28, no.~5,
  pp. 814--816, 1982.

\bibitem{kumar1995upper}
P.~V. Kumar, T.~Helleseth, and A.~R. Calderbank, ``An upper bound for {Weil}
  exponential sums over {Galois} rings and applications,'' \emph{IEEE
  transactions on Information Theory}, vol.~41, no.~2, pp. 456--468, 1995.

\end{thebibliography}

\appendix
\label{appendix}

\begin{lemma}
\label{dist-quasidist-relation}
    Given a metric $d_0:M \times M \xrightarrow{} \R$ on a set $M$, the function $d(x,y)=d_0(x,y)^2:M \times M \xrightarrow{} \R$ is a $2$-quasimetric on $M$.
\end{lemma}
\begin{proof}
 We only need to show that \eqref{quasi-triangle-ineq} holds for $d(.,.)$ with $\sigma=2$. The proof is by noting that for any $x,y,z \in M$ we have
 \begin{align}
     d(x,z)=d_0(x,z)^2 &\leq (d_0(x,y)+d_0(y,z))^2 \label{tri-ineq}\\
     &\leq 2(d_0(x,y)^2+d_0(y,z)^2)\label{non-neg}\\
     &=2(d(x,y)+d(y,z)),
 \end{align}
 where \eqref{tri-ineq} follows from the triangle inequality and \eqref{non-neg} is by Cauchy-Schwarz inequality. The remaining properties of a quasimetric follow in a straightforward way given that $d_0$ is a metric. 
\end{proof}

\begin{lemma}
\label{sim_diag}
    Suppose that the projection matrices of two subspaces $U,V \in \cP(W)$ are simultaneously diagonalizable. Then the distance $d(.,.)$, as defined in \eqref{dist-def}, satisfies the triangle inequality, i.e., 
    $$
    d(U,V)+d(V,T)\geq d(T,U),
    $$
    for any $T \in \cP(W)$.
\end{lemma}
\begin{proof}
     By definition of simultaneously diagonalizable matrices, there exists an orthonormal basis for the ambient space $W$ in which both $\bP_U$ and $\bP_V$ are diagonal. Let $u_{i,j}$, $v_{i,j}$ and $t_{i,j}$, for $i,j \in [n]$, denote the entries of the projection matrices $\bP_U$, $\bP_V$, and $\bP_T$ in this basis, respectively. Note that $u_{i,j}=v_{i,j}=0$ for $i\neq j$. Also, the diagonal entries of $\bP_U$ and $\bP_V$ are either $0$ or $1$, since $\bP_U$ and $\bP_V$ are projection matrices. Let 
     $$
     I_U\deff\{i: u_{i,i}=1, i\in [n]\},
     $$ 
     and 
     $$
     I_V\deff\{i: v_{i,i}=1, i\in [n]\}. 
     $$
     Also, we have $0\leq t_{i,i}\leq 1$, for  $i \in [n]$, a property that holds for diagonal entries of any projection matrix $\bP_T$. By \eqref{d function} we have the following relations for the pairwise distances between $U$, $V$, and $T$:
\begin{align}
        \label{duv}
        &d(U,V)=|I_U \cap I_V^c|+|I_U^c \cap I_V|,\\
        \label{dtu}
        &d(T,U)=\sum_{i,j\in [n]} |t_{i,j}|^2+\sum_{i\in I_U}(1-2t_{i,i}),\\
        \label{dvt}
       & d(V,T)=\sum_{i,j \in [n]} |t_{i,j}|^2+\sum_{i\in I_V}(1-2t_{i,i}),
\end{align}
where the complements in \eqref{duv} are taken with respect to $[n]$. Subtracting \eqref{dvt} from \eqref{dtu} yields
\begin{align}
d(T,U)-d(V,T)&=\hspace{-2mm}\sum_{i\in I_U\cap I_V^c}\hspace{-2mm}(1-2t_{i,i})-\hspace{-2mm}\sum_{i\in I_U^c\cap I_V}\hspace{-2mm}(1-2t_{i,i})\\
\label{ineq}
&\leq |I_U \cap I_V^c|+|I_U^c \cap I_V|,
\end{align}
where \eqref{ineq} is by noting that $-1\leq 1-2t_{i,i}\leq1$ for $i\in [n]$. This together with \eqref{duv} complete the proof.
\end{proof}

\begin{lemma}
\label{orthogonal subspaces distance mirroring}
For any $U,V\in \cP(W)$ we have
\be
\label{orthogonal distance}
d(U^\perp,V^\perp)=d(U,V),
\ee 
where $d(.,.)$ is defined in Definition\,\ref{dist-def}.
\end{lemma}
\begin{proof}
It is well known that $\bP_{U^\perp}=\bI-\bP_{U}$. Hence,
$$
d(U^\perp,V^\perp)\hspace{-1mm}=\hspace{-1mm}\text{tr}((\bP_U^\perp-\bP_V^\perp)^2)\hspace{-1mm}=\hspace{-1mm} \text{tr}((\bP_V-\bP_U)^2)\hspace{-1mm}=\hspace{-1mm} d(U,V).
$$
\end{proof}

\begin{lemma}
\label{direct sum related to distance}
Suppose that $U,T \in \cP(W)$ and let $V=U \oplus T$. Then we have
\be
\label{direct sum distance }
d(U,V)=\dim(T),
\ee
where $d(.,.)$ is defined in Definition\,\ref{dist-def}.
\end{lemma}
\begin{proof}
Let $t=\dim{T}$ and $u=\dim{U}$. Then $\dim{V}=t+u$. One can always find a basis for $W$, namely $\{\boldsymbol{e}_1,\hdots,\boldsymbol{e}_n\}$, such that $U=\Span{\bI_u}$ and $V=\Span{\bI_v}$ where $\bI_u$ and $\bI_v$ are identity matrices of dimensions $u$ and $v$, respectively. Consequently, the orthogonal projection matrices associated with these subspaces are $\bP_U=\left[
\begin{array}{c|c}
\bI_u & \boldsymbol{0} \\
\hline
\boldsymbol{0} & \boldsymbol{0}
\end{array}
\right]$ and $\bP_V=\left[
\begin{array}{c|c}
\bI_v & \boldsymbol{0} \\
\hline
\boldsymbol{0} & \boldsymbol{0}
\end{array}
\right]$. 
Then the lemma follows immediately  by using \eqref{d function} and noting that the distance is rotation invariant by \Lref{rotation invariant}.
\end{proof}

\begin{lemma}\label{BHB}
    Let $\bB\in \C^{l\times n}$. Then $ \norm{
    \bB^{\text{H}}\bB}\leq \norm{\bB}^2.$

\end{lemma}
\begin{proof}
    Let $\sigma_i \in \R$, for $i \in [l]$, denote the singular values of $\bB$.  Then the singular values of $\bB^{\text{H}}\bB $ are $\lambda_i^2$'s. Hence, by using \eqref{Frobnius norm deff} we have
    $$
    \norm{\bB\bB^{\text{H}}}=\sqrt{\sum_{i=1}^l \sigma_i^4} \leq \sum_{i=1}^l \sigma_i^2=\norm{\bB}^2.
    $$
\end{proof}

\end{document}